\documentclass[11pt]{article}

\usepackage[margin=1in]{geometry}
\usepackage{libertine}
\usepackage[T1]{fontenc}
\usepackage{microtype}
\usepackage{amsmath, amsthm, mathtools}
\usepackage[libertine]{newtxmath}

\usepackage{xcolor}

\definecolor{linkblue}{HTML}{155195}
\definecolor{citepurple}{HTML}{721679}
\usepackage[pdfstartview=FitH,colorlinks,linkcolor=linkblue,filecolor=linkblue,citecolor=citepurple,urlcolor=linkblue,pagebackref]{hyperref}
\hypersetup{
  pdftitle={The Head Complexity of Boolean Functions in Single-Layer Attention},
  pdfauthor={Rajmohan Rajaraman, Ravi Sundaram, Amanuel Tesfaye},
  pdfsubject={Computational complexity of single-layer attention},
  pdfkeywords={transformers, attention heads, head complexity, parity, Boolean functions, circuit complexity, expressivity}
}
\usepackage{cleveref}

\theoremstyle{plain}
\newtheorem{theorem}{Theorem}
\newtheorem{lemma}{Lemma}
\newtheorem{claim}{Claim}

\theoremstyle{definition}
\newtheorem{definition}{Definition}

\theoremstyle{remark}
\newtheorem{remark}{Remark}

\title{\sffamily\bfseries The Head Complexity of Boolean Functions\\
in\\
Single-Layer Attention}
\author{
  Rajmohan Rajaraman\thanks{Northeastern University.}
  \and Ravi Sundaram\footnotemark[1]
  \and Amanuel Tesfaye\footnotemark[1]
}
\date{}

\begin{document}

\pagenumbering{gobble}
\maketitle
\thispagestyle{empty}

\begin{abstract}
What can a single layer of self-attention compute? We study \emph{head complexity:} the minimum number of attention heads required to compute a function in a one-layer attention-only model. We establish an exact hierarchy under this measure: $k$ heads compute $k$-bit parity but cannot compute $(k+1)$-bit parity.

The lower bound is unconditional in the two resources a transformer might otherwise exploit; it holds at unbounded embedding dimension and unbounded numerical precision. The proof rests on an alternating-sum obstruction: after clearing the softmax denominators, every monomial in the resulting decision polynomial omits at least one of the $k+1$ input bits, forcing its correlation with parity to vanish. The same obstruction yields lower bounds for related tasks, including the well-studied multi-hop induction-head task.

\smallskip

We also establish compactness bounds for embedding dimension and numerical precision. Specifically, a compactness theorem shows that any function computable at all can be computed with embedding dimension and precision bounded by the discrete data of the task, namely, head count, alphabet size, and length. Thus, potentially unbounded dimension or precision provably cannot substitute for heads. Finally, we derive nearly matching universal bounds for general binary functions: $2^n$ heads suffice to compute every $n$-bit binary function, with one head per monomial in its multilinear expansion, while a counting argument shows almost all such functions require $\Omega(2^n/n^2)$ heads. This lower bound matches the upper bound to within a $\operatorname{poly}(n)$ factor, even when dimension and precision are unbounded. Together, these results characterize head requirements for Boolean computation in this model.

\end{abstract}

\clearpage
\pagenumbering{arabic}
\section{Introduction}
\label{sec:intro}

To understand what a transformer can compute, we analyze it not as an empirical artifact but as a formal model of computation. Half a century ago, the pioneers of circuit complexity asked how many logic gates a function requires, and the structural limits they proved for bounded-depth circuits became bedrock for the field. For a single attention layer, the natural analogue of the gate is the \emph{attention head}, and the natural question is its \emph{head complexity}: the minimum number of heads needed to compute a target function. This lens grew out of the two-head study~\cite{tesfaye2026twoheads}, where two \emph{narrow} heads solve the Endpoint Selection Problem (ESP) that a single head of \emph{unbounded} width cannot; this motivated head count as a resource for studying single-layer expressivity. The present paper extends that observation to an exact hierarchy across head counts and to worst-case bounds for Boolean functions.

Why isolate a \emph{single} attention layer? We isolate a single attention layer in order to study head count separately from depth and feed-forward computation. This restricted setting permits exact upper and lower bounds. Settling the head complexity has two additional payoffs.  It is a statement about cost --- heads carry parameters and quadratic attention compute, so a lower bound on heads lower-bounds model size within this attention-only setting. And it grounds mechanistic interpretability, whose unit of analysis is the individual head. We therefore take the number of heads as the resource being measured.

\subsection{Summary of Results}

We measure a self-attention layer by its number of heads. Our main results are the following.

\begin{itemize}
\item \textbf{An exact head hierarchy, calibrated by parity.} We prove that $k$ heads cannot compute $(k{+}1)$-bit parity, and give a matching construction in which $k$ heads compute $k$-bit parity. Hence $m$-bit parity needs \emph{exactly} $m$ heads, and the head-count classes form a strictly increasing chain $\mathcal{F}_1\subsetneq\mathcal{F}_2\subsetneq\cdots$ that separates at every single head. The lower bound is unconditional: it holds at unbounded embedding dimension and unbounded numerical precision.

\item \textbf{One obstruction, many tasks.} The lower bound is an elementary alternating-sum obstruction, and it is reusable: the same obstruction gives the exact $(k{+}1)$-head requirement for $k$-ESP, which is a generalization of the ESP task introduced in~\cite{tesfaye2026twoheads} and an unconditional $(k{+}1)$-head lower bound for the $(k{+}1)$-hop induction-head task~\cite{sanford2024logdepth}. The hierarchy is thus intrinsic to one-layer attention, not a peculiarity of parity.

\item \textbf{Dimension and precision can be bounded.} A compactness theorem shows that any computable function has a realization whose embedding dimension and  numerical precision are bounded in terms of the discrete data of the task: head count, alphabet size, and sequence length. The two potentially unbounded resources provably cannot substitute for heads.

\item \textbf{How many heads the hardest functions need.} With $2^n$ heads a single layer computes \emph{every} $n$-bit Boolean function (one additive head per monomial of its multilinear expansion); a counting argument shows that almost all functions require $\Omega(2^n/n^2)$ heads. The two bounds meet to within a $\mathrm{poly}(n)$ factor, even with unbounded embedding dimension and numerical precision.
\end{itemize}

Together, these results show that head count strictly orders expressivity in the one-layer attention-only model, while unbounded dimension and precision do not eliminate the need for additional heads. Throughout, our upper bounds use the restricted \emph{additive} embedding and our lower bounds the general \emph{joint} one, the strongest pairing possible for these two regimes, since it pins both at once. One further result appears in the appendix:  any joint $k$-head computation can be simulated additively with $O(n^k)$ heads, giving polynomial overhead in $n$ for every fixed $k$.

\subsection{Closest Prior Work}
\label{subsec:closest-prior}

The nearest parity lower bounds are due to Kozachinskiy, Steifer, and Wa\l{}\k{e}ga~\cite{kozachinskiy2026parity} and Hsu~\cite{hsu2026parity}. The former proves the qualitative statement that a fixed one-layer softmax transformer cannot compute parity, even with an MLP, using sensitivity and first-order-logic arguments; the latter proves the quantitative bound $hp \ge n/4$ for one-layer softmax attention with a degree-$p$ rational read-out, which for a linear read-out gives $h=\Omega(n)$ heads. Both establish that parity is hard for one layer, but neither gives an exact head threshold. Our lower bound rules out $k$ heads for $(k{+}1)$-bit parity while our construction attains $k$-bit parity with $k$ heads, so the hierarchy separates at every single head; the compactness theorem then shows that allowing unbounded dimension or precision does not enlarge the class of computable functions. A full discussion of related work appears in \Cref{sec:related}.

\subsection{Overview of Techniques}

Our results rest on four main ideas.

\paragraph{The obstruction (\Cref{sec:parity-lb}).} We first argue that clearing the strictly positive softmax denominators that appear in transformer computations turns the decision of a $k$-head transformer into the sign of a single polynomial $F_x$ in the effective head variables. Each monomial of $F_x$ is a product of $k$ factors, and each factor reads only one input position, so a monomial depends on at most $k$ of the $k+1$ parity bits: some bit is always missing. Summing $F_x$ against the parity sign then factors through $\sum_{b\in\{0,1\}}(-1)^b=0$, giving $\sum_x(-1)^{f(x)}F_x=0$; this is impossible if every signed summand is nonnegative and positive for odd-parity inputs, as computing parity would require. The argument invokes nothing about the embedding, which is why it survives at unbounded dimension and precision, and nothing about parity beyond ``each token sees one bit,'' which is why it transfers to $k$-ESP and $(k{+}1)$-hop through a suitable encoding.

\paragraph{The interpolation (\Cref{sec:parity-ub}).} The matching construction runs the rational-function viewpoint in reverse. The heads are arranged so that, on an input of Hamming weight $c$, head $j$ has denominator $c+j$ and numerator gap affine in $c$, hence contributes a term $(\alpha_j+\beta_j c)/(c+j)$ to the decision statistic. A partial-fraction choice of the coefficients makes the $k$-term sum alternate in sign on the $k+1$ Hamming levels $c=0,\ldots,k$, exactly matching parity. The $k$-ESP upper bound also uses the same Hamming-weight interpolation technique.

\paragraph{Compactness (\Cref{sec:compactness}).} The layer is a bilinear machine: the function it computes reads its parameters only through one rectangular matrix of inner products, pairing the query and read-out vectors against the token embeddings.  The required embedding dimension is at most one greater than the rank of that matrix. For precision, Renegar-type estimates~\cite{renegar1992} first produce a bounded rational point in exponentiated-score coordinates; a quantitative stability argument transfers this bound to rational scores and hence to the underlying model matrices.

\paragraph{Counting (\Cref{sec:univ-lb}).} The entire input--output behavior of a $k$-head transformer is the sign vector of the $2^n$ polynomials $\{F_x\}$ evaluated at a single point in $O(kn)$ variables. Warren's bound on the number of sign patterns of low-degree polynomials caps the number of realizable functions at $2^{O(kn^2)}$; comparing this with the $2^{2^n}$ Boolean functions forces $k=\Omega(2^n/n^2)$ for almost all of them. The nearly matching universal upper bound is a direct construction: one additive head per monomial of the target's multilinear expansion.

\section{Preliminaries}
\label{sec:prelim}

We study a one-layer, attention-only transformer that reads a Boolean input and predicts one output bit. We state the standard two-logit model and immediately reduce its decision to a single scalar. For some arguments we use an affine scalar threshold; \Cref{thm:readout-equiv} shows that this changes dimension by at most one coordinate but never changes the number of heads.

\subsection{Inputs and embeddings}

Fix $n$ and present $x=(x_1,\ldots,x_n)\in\{0,1\}^n$ as the sequence $x_1x_2\cdots x_n\,?$, where $? $ is a fixed query token. Write $v_{b,i}:=\Phi(b,i)\in\mathbb{R}^d$ for $b\in\{0,1\}$ and $i\in[n]$, and $v_?:=\Phi(?)\in\mathbb{R}^d$. Let $X=X(x)\in\mathbb{R}^{(n+1)\times d}$ have rows $v_{x_1,1},\ldots,v_{x_n,n},v_?$. In the general \emph{joint} regime the vectors $v_{b,i}$ are unrestricted; the \emph{additive} regime requires $v_{b,i}=\tau_b+\pi_i$. Thus every additive model is a joint model.

\subsection{The model}

\begin{definition}[One-layer $k$-attention-only transformer, arg-max read-out]
\label{def:k-attn}

Fix an embedding dimension $d$, a value width $d_{\mathrm{out}}$, and $k$ heads. Head $j$ has $A_j\in\mathbb{R}^{d\times d}$ and $V_j\in\mathbb{R}^{d\times d_{\mathrm{out}}}$; the two output logits share $W_O\in\mathbb{R}^{2\times d_{\mathrm{out}}}$ with rows $w_0,w_1$. At the query position,
\[
\mathrm{MHA}(x):=\sum_{j=1}^{k}\mathrm{Softmax}(v_?A_jX^{\top})XV_j,
\qquad
Z_s(x):=\langle w_s,\mathrm{MHA}(x)\rangle \quad(s\in\{0,1\}).
\]
The model sets $\mathrm{out}(x):=\arg\max_{s\in\{0,1\}}Z_s(x)$, resolving ties toward $0$; equivalently, $\mathrm{out}(x)=1$ exactly when $Z_1(x)>Z_0(x)$. It computes $f:\{0,1\}^n\to\{0,1\}$ when $\mathrm{out}(x)=f(x)$ for every $x$.
\end{definition}

\subsection{Scalar normal form}

Only the difference between the two logits matters. For head $j$, set
\[
\begin{aligned}
\delta_j&:=V_j(w_0-w_1)^\top,&
r^j_{b,i}&:=e^{v_?A_jv_{b,i}^\top},&
u^j_{b,i}&:=v_{b,i}\cdot\delta_j,\\
r^j_?&:=e^{v_?A_jv_?^\top}>0,&
u^j_?&:=v_?\cdot\delta_j,\\
d_x^j&:=r^j_?+\sum_{i=1}^nr^j_{x_i,i}>0,&
s_x^j&:=r^j_?u^j_?+\sum_{i=1}^nr^j_{x_i,i}u^j_{x_i,i}.
\end{aligned}
\]
Writing $a_j,b_j$ for the two columns of the folded map $Y_j:=V_jW_O^\top$, the vector numerator and scalar bridge satisfy
\[
n_x^j:=r_?^jv_?+\sum_{i=1}^nr_{x_i,i}^jv_{x_i,i},\qquad
\delta_j=a_j-b_j,\qquad s_x^j=n_x^j\cdot(a_j-b_j),\qquad
\frac{n_x^j}{d_x^j}\cdot(a_j-b_j)=\frac{s_x^j}{d_x^j}.
\]
Thus the entire Boolean decision is governed by the output-$0$ gap
\begin{equation}\label{eq:prelim-computes}
  \Delta(x):=Z_0(x)-Z_1(x)=\sum_{j=1}^{k}\frac{s_x^j}{d_x^j},
  \qquad
  \mathrm{out}(x)=1\iff\Delta(x)<0.
\end{equation}
Hence the transformer computes $f$ exactly when $f(x)=1\iff\Delta(x)<0$ for every input. This ratio-of-sums form is the only reduction used in the main arguments; a proof clears denominators only when needed.

\begin{definition}[Affine scalar read-out]
\label{def:ltf}
An affine scalar read-out uses the attention core of \Cref{def:k-attn} and outputs
\[
  \mathrm{net}(x)=\mathbf 1\{\langle\omega,\mathrm{MHA}(x)\rangle-\theta>0\},
  \qquad \omega\in\mathbb{R}^{d_{\mathrm{out}}},\ \theta\in\mathbb{R}.
\]
\end{definition}
For constructions and counting we occasionally use this convention. As stated and proved in \Cref{thm:readout-equiv}, it computes exactly the same Boolean functions at every head count as the two-logit model above.

\section{The Exact Parity Hierarchy}
\label{sec:hierarchy}

The resource we track is the number of attention heads. This section shows that it strictly, and exactly, orders the power of the one-layer model of \Cref{def:k-attn}: for every $k$, one more head lets the network compute strictly more Boolean functions, and $(k{+}1)$-bit parity marks the boundary. We prove the lower bound first (\Cref{sec:parity-lb}), then the matching upper bound (\Cref{sec:parity-ub}).

\paragraph{Head-count classes.}
Let $\mathcal F_k=\mathcal F_k^{\mathrm{joint}}$ denote the Boolean functions, of any arity, computable by at most $k$ heads in the joint model, with all other resources unbounded; let $\mathcal F_k^{\mathrm{add}}$ denote the additive counterpart. Padding with zero-valued heads gives $\mathcal F_k\subseteq\mathcal F_{k+1}$. Writing $\mathrm{PAR}_m$ for $m$-bit parity, the two results below show that every inclusion is strict:
\[
  \mathrm{PAR}_{k+1}\in\mathcal F_{k+1}\setminus\mathcal F_k.
\]

\subsection{The alternating-sum obstruction: \texorpdfstring{$k$}{k} heads cannot compute \texorpdfstring{$(k{+}1)$}{(k+1)}-bit parity}
\label{sec:parity-lb}

The $(k{+}1)$-parity problem takes input $x=(x_1,\ldots,x_{k+1})\in\{0,1\}^{k+1}$, presented as the $(k{+}2)$-token string $x_1 x_2 \cdots x_{k+1}\, ?$, and asks for the bit $\bigoplus_{1 \le i \le k+1} x_i$.

We specialize the computation criterion \eqref{eq:prelim-computes} of \Cref{sec:prelim} to this task. Recall from there that, since the query may attend to itself, head $j$ contributes the numerator $n_x^j = r^j_?\,v_? + \sum_{1 \le i \le k+1} r^j_{x_i, i}\, v_{x_i, i}$ and the denominator $d_x^j = r^j_? + \sum_{1 \le i \le k+1} r^j_{x_i, i}>0$, where $v_{b,i}=\Phi(b,i)$, $v_?=\Phi(?)$, $r^j_{b,i}=e^{v_?A_j v_{b,i}^\top}>0$, and $r^j_?=e^{v_?A_jv_?^\top}>0$ is the query's (input-independent) self-weight; $a_j,b_j$ are the two columns of the folded read-out $Y_j=V_jW_O^\top$ (for outputs $0$ and $1$). By Equation~\ref{eq:prelim-computes}, a one-layer $k$-head transformer solves $(k{+}1)$-parity iff
\[
(-1)^{\sum_{1 \le i \le k+1} x_i} \cdot \sum_{j=1}^{k} \frac{n^j_x}{d^j_x} \cdot (a_j - b_j) \;\ge\; 0
\qquad\text{for every } x \in \{0,1\}^{k+1},
\]
with strict inequality whenever the parity is 1.

\begin{theorem}
\label{thm:parity-lb}
There is no choice of vectors $v_?,v_{0,i}, v_{1,i}, a_j, b_j$, and positive reals $r^j_{0,i}, r^j_{1,i}, r^j_?$, for
$1 \le i \le k+1$, $1 \le j \le k$, such that for every
$x \in \{0,1\}^{k+1}$ the following holds:
\[
(-1)^{\sum_{1\le i \le k+1} x_i} \cdot \sum_{j=1}^{k} \frac{n^j_x}{d^j_x} \cdot (a_j - b_j)
\;\begin{cases}
> 0 & \text{if } \sum_{1 \le i \le k+1} x_i \text{ is odd (target bit }=1),\\[2pt]
\ge 0 & \text{if } \sum_{1 \le i \le k+1} x_i \text{ is even (target bit }=0).
\end{cases}
\]
\end{theorem}

\begin{proof}

Recall that $\sigma_x=(-1)^{\sum_{i=1}^{k+1}x_i}$.

\medskip

\noindent\textbf{Step 1: Clear the denominators.}\quad Set
\[
s_x^j:=n_x^j\cdot(a_j-b_j)
=r_?^j\bigl(v_?\cdot(a_j-b_j)\bigr)
+\sum_{i=1}^{k+1}r^j_{x_i,i}\bigl(v_{x_i,i}\cdot(a_j-b_j)\bigr).
\]
Each of $s_x^j$ and $d_x^j$ is therefore a sum of terms depending on at most one input bit: the query self-terms are constant, while the $i$th summands depend only on $x_i$. Define $\Delta(x):=\sum_{j=1}^k s_x^j/d_x^j$.

The inequality at $x$ reads $\sigma_x\Delta(x)\ge0$, with strict inequality whenever $\sigma_x=-1$.
Since each $d^j_x > 0$, multiplying through by the strictly positive quantity $\prod_{j=1}^{k} d^j_x$ preserves sign:

\[
\mathrm{sign}(\Delta(x))=\mathrm{sign}(F_x),\qquad
F_x:=\prod_{j=1}^k d_x^j\,\Delta(x)=\sum_{j=1}^k s_x^j\prod_{j'\ne j}d_x^{j'}.
\]
The theorem is therefore equivalent to the assertion: no choice of parameters makes $\sigma_x F_x \ge 0$ hold for every $x \in \{0,1\}^{k+1}$, with strict inequality whenever $\sigma_x = -1$.

\noindent\textbf{Step 2: Alternating-sum identity.}
\begin{claim}
For every choice of the parameters,
$\sum_{x \in \{0,1\}^{k+1}} \sigma_x F_x=0$.
\end{claim}

\begin{proof}[Proof of claim]
Expand $F_x$ as a sum of monomials.  Each term in the defining sum has the form $s^j_x \prod_{j' \ne j} d^{j'}_x$, and we further expand each factor using

\[
\begin{aligned}
s_x^j&=r_?^j\bigl(v_?\cdot(a_j-b_j)\bigr)
+\sum_{i=1}^{k+1}r^j_{x_i,i}\bigl(v_{x_i,i}\cdot(a_j-b_j)\bigr),\\
d_x^{j'}&=r^{j'}_?+\sum_{i=1}^{k+1}r^{j'}_{x_i,i}.
\end{aligned}
\]

A typical monomial $m$ is a product of exactly $k$ factors---one from $s_x^j$ and one from each of the $k-1$ remaining denominators. Each factor is either a query constant ($r_?^j(v_?\cdot(a_j-b_j))$ or $r^{j'}_?$) or a single-position term ($r^j_{x_i,i}(v_{x_i,i}\cdot(a_j-b_j))$ or $r^{j'}_{x_i,i}$).
Hence $m$ is a product of $k$ scalars, each depending on \emph{at most} one binary coordinate, so $m$ depends on at most $k$ of the $k+1$ coordinates $x_1, \ldots, x_{k+1}$.  Since $k < k+1$, there exists at least one index $i^* \in \{1, \ldots, k+1\}$ on which $m$ does not depend.  Splitting the sum over $x$ by isolating the free coordinate $x_{i^*}$,
\[
\sum_{x \in \{0,1\}^{k+1}} \sigma_x \cdot m \;=\; \Bigg(\sum_{\substack{x_l \in \{0,1\} \\ l \ne i^*}} (-1)^{\sum_{l \ne i^*} x_l}\, m\Bigg) \cdot \underbrace{\sum_{x_{i^*} \in \{0,1\}} (-1)^{x_{i^*}}}_{= 0} \;=\; 0.
\]
Since every monomial of $F_x$ vanishes individually under the alternating sum, so does: $\sum_x \sigma_x F_x = 0$.
\end{proof}

\noindent\textbf{Step 3: Contradiction.}\quad
Suppose, toward a contradiction, that the required inequalities hold simultaneously: for every $x \in \{0,1\}^{k+1}$, $\sigma_x F_x \;\ge\; 0$,
with strict inequality whenever $\sigma_x = -1$ (equivalently, whenever the target bit is $1$).
Since parity is balanced on $\{0,1\}^{k+1}$, there are $2^k > 0$ inputs with odd Hamming weight,
so at least one term $\sigma_x F_x$ is strictly positive; every other term is nonnegative by
hypothesis. Hence
\[
\sum_{x \in \{0,1\}^{k+1}} \sigma_x F_x \;>\; 0,
\]
contradicting the identity $\sum_x \sigma_x F_x = 0$ established in Step 2. Therefore no choice of
parameters can satisfy all $2^{k+1}$ inequalities, completing the proof.

\end{proof}

\begin{remark}\label{rem:parity-obstruction}
The same alternating-sum argument establishes a more general \emph{parity obstruction}.  Consider any system of $2^P$ inequalities indexed by $i = (i_1, \ldots, i_P) \in \{0,1\}^P$, with required signs $(-1)^{i_1 + \cdots + i_P}$, of the form
\[
(-1)^{i_1 + \cdots + i_P} \cdot \sum_{j=1}^{H} \frac{s^j_i}{d^j_i} \;>\; 0,
\]
in which each $s^j_i$ and each $d^j_i > 0$ is a sum of real-valued summands, each depending on at most one of the binary indices.  Whenever $P > H$, the system is infeasible: defining $F_i = \sum_j s^j_i \prod_{j' \ne j} d^{j'}_i$, each monomial in $F_i$ is a product of $H$ such summands and therefore involves at most $H < P$ of the binary indices, leaving at least one free index to kill the alternating sum, and the same contradiction in Step 3 closes the argument.  The condition $P > H$ is tight: at $P = H$, a single monomial of $F_i$ may collectively involve all $P$ binary indices.

In the present setting, $P = k+1$ binary indices (the input bits) and $H = k$ heads, so $P > H$ by exactly one and the theorem follows.
\end{remark}

\paragraph{Beyond parity.}
The lower bound extends to other problems that meet the hypotheses of \Cref{rem:parity-obstruction}: $k$ heads cannot solve the $k$-ESP endpoint-selection task of~\cite{tesfaye2026twoheads}, which $k{+}1$ heads solve exactly (\Cref{sec:k-esp-lb,sec:k-esp-ub}), nor the $(k{+}1)$-hop induction-heads task of~\cite{sanford2024logdepth} (\Cref{sec:khop}).

\subsection{Matching upper bound: \texorpdfstring{$k$}{k} heads compute \texorpdfstring{$k$}{k}-bit parity}
\label{sec:parity-ub}

We now match the lower bound. For $k$-bit parity the input is the $(k{+}1)$-token string $x_1x_2\cdots x_k\,?$ with target $\bigoplus_{i=1}^k x_i$; by the criterion \eqref{eq:prelim-computes} of \Cref{sec:prelim}, to prove the upper bound, it suffices to construct parameters such that

\[
(-1)^{\sum_{i=1}^k x_i}
\underbrace{\sum_{j=1}^{k}\frac{n_x^j}{d_x^j}\cdot(a_j-b_j)}_{\displaystyle
=\,\Delta(x)=\sum_{j=1}^{k}s_x^j/d_x^j}>0
\qquad\text{for every }x\in\{0,1\}^k .
\]
The lower bound (\Cref{thm:parity-lb}, $k-1$ heads against $k$ bits) rules out $k-1$ heads; we show $k$ suffice, even in embedding dimension $d=k+3$.

\begin{theorem}\label{thm:parity-ub}
For every $k\ge1$ there is a choice of embedding $\Phi$, matrix $W_O$, and attention and value matrices $A_j,V_j$ such that the one-layer $k$-head attention-only transformer above solves $k$-bit parity.
\end{theorem}

\medskip
\noindent\textbf{Proof sketch.}\quad The proof has three main ideas. First, make each head's denominator equal to $c+j$, where $c=\sum_{i=1}^k x_i$ is the Hamming weight of the input. Second, use the value matrices $V_j$ to make each head's numerator gap equal to an affine function $\alpha_j+\beta_jc$. Third, choose the coefficients so that the sum $\Delta(c)=\sum_{j=1}^k(\alpha_j+\beta_jc)/(c+j)$ has sign $(-1)^c$ for every $c=0,\ldots,k$. Since the required signed gap is $(-1)^c$, this solves parity.

\noindent\emph{The supporting partial-fractions lemma and the full construction are deferred to \Cref{app:deferred-proofs}.}

\begin{remark}
The head count is tight. The matching lower bound rules out $k-1$ heads for $k$-parity, and the construction attains $k$ using an additive token-plus-position embedding. The value/output gap gives different weights to token $0$ and token $1$, so each head has an affine numerator in the total Hamming weight $c$. Summing $k$ affine-over-linear terms gives a rational function whose sign alternates on all $k+1$ Hamming levels $0,\ldots,k$.
\end{remark}

\section{Compactness: Dimension and Precision Are Bounded by the Discrete Data}
\label{sec:compactness}

The lower bounds of this paper charge the model for \emph{heads}. A natural question to ask is whether for a given number of heads, the capabilities of the transformer can greatly increase with the power offered by the other two continuous resources: say, an enormous embedding dimension $d$, or unbounded arithmetic precision. This section rules that out.  We show that if a function is computable by the one-layer $k$-attention-only model at all, it is computable with \emph{both} resources bounded by the \emph{discrete} data of the problem: the head count $k$, the number of output labels $L$, the number of token values $q=|\Sigma|$, and the sequence length $N$. The key idea behind this compactness result is a single observation: the model is a \emph{bilinear (inner-product) machine}, and the function it computes reads its parameters off one rectangular matrix of inner products.

The argument has two independent parts. First, the computation factors through a finite cross-Gram matrix, whose rank bounds the required dimension. Second, correctness is described by finitely many strict polynomial inequalities; quantitative rational approximation and stability then bound the bit-length of a realizing model.

We state the results for the model of \Cref{def:k-attn} generalized to a value alphabet $\Sigma$ ($q=|\Sigma|$) and $L$ output labels (output matrix $W_O\in\mathbb{R}^{L\times d_{\mathrm{out}}}$, folded $Y_j=V_jW_O^{\top}\in\mathbb{R}^{d\times L}$ with columns $y^{(1)}_j,\dots,y^{(L)}_j$; in the binary case, $q=L=2$, $N=n+1$, and the columns are indexed by labels $0$ and $1$). Write $\xi_j=v_?A_j$ for the effective query, $p^{j}_t=\xi_j\cdot v_t$ for the score of token type $t$ (an occurring (value, position) pair, with embedding $v_t$), and $w^{j,l}_t=v_t\cdot y^{(l)}_j$ for its value-projection to label $l$.

For a rational number $z=a/b$ in lowest terms, let its bit-length be the maximum of the binary lengths of $a$ and $b$.  The numerical precision of a rational model realization is the maximum bit-length of an entry of its embedding table and weight matrices $\Phi,\{A_j,V_j\}_{j\in[k]},W_O$.

\subsection{The computation factors through a cross-Gram matrix}

\begin{lemma}[Factoring lemma]
\label{lem:factoring}
Fix a reference label $\star$ and set
\[
  \mathcal{A}=\{\xi_j\}_{j\in[k]}\ \cup\ \{\,y^{(l)}_j-y^{(\star)}_j\,\}_{j\in[k],\,l\ne\star},
  \qquad
  \mathcal{B}=\{\,v_t : t\text{ an occurring token type}\,\},
\]
so $|\mathcal{A}|=kL$ and $|\mathcal{B}|\le qN$. The function computed depends on the parameters $\Phi,\{A_j\},\{V_j\},W_O$ \emph{only} through the cross-inner-product matrix
\[
  M\in\mathbb{R}^{\mathcal{A}\times\mathcal{B}},\qquad M_{a,t}=a\cdot v_t .
\]
No token--token product $v_t\cdot v_{t'}$, and no product within $\mathcal{A}$, ever enters.
\end{lemma}
\begin{proof}
Only logit differences decide the prediction. Each is built from the softmax weights and the weighted value sums over the token types occurring in the input; the former are functions of the scores $p^{j}_t=\xi_j\cdot v_t$ alone, the latter of $w^{j,l}_t-w^{j,\star}_t=(y^{(l)}_j-y^{(\star)}_j)\cdot v_t$. Both families are exactly the entries of $M$.
\end{proof}

\begin{remark}[Why the model is ``bilinear'']
The lemma says the geometry the function can see is one-sided: it reads how each parameter vector $a\in\mathcal{A}$ pairs against each token embedding $v_t\in\mathcal{B}$, but never how the $v_t$ sit relative to one another, nor how the $a$ do. Dimension and precision are properties of that ambient geometry; the function only reads a single rectangular slice of it whose size depends on other model parameters; hence unbounded dimension or precision does not add more to the expressivity of the transformer.
\end{remark}

\subsection{Dimension: if realizable, realizable in dimension \texorpdfstring{$\operatorname{rank}(M)+1$}{rank(M) + 1}}

\begin{theorem}[Dimension compactness]
\label{thm:dim-compact}
Every realization has $d \ge \operatorname{rank}(M)$, and any prescribed cross-Gram matrix $M$ is
realizable in dimension $r+1$, where $r=\operatorname{rank}(M)$. Hence the least embedding
dimension computing $f$ satisfies
\[
\operatorname{rank}(M^\star) \;\le\; d^\star(f) \;\le\; \operatorname{rank}(M^\star)+1,
\qquad
M^\star \in \operatorname*{arg\,min}_{M \text{ realizes } f} \operatorname{rank}(M),
\]
and in particular $d^\star(f) \le \min\{kL,\,qN\}+1$; for binary output,
$d^\star \le \min\{2k,\,qN\}+1$.
\end{theorem}

\noindent\emph{Proof deferred to \Cref{app:deferred-proofs}.}

\begin{remark} For the tasks studied in the paper (binary $L=2$, $q=2$, $N=\Theta(k)$) Theorem~\ref{thm:dim-compact} yields $d^{\star}\le\min(2k,2N)+1=O(k)$, matching the linear bound on the dimension used by all constructions in the paper.
\end{remark}

\subsection{\texorpdfstring{Precision: if realizable, realizable with bounded-bit rational  model parameters}{Precision: if realizable, realizable with bounded-bit rational model parameters}}

To expose polynomial constraints, introduce the auxiliary variables $r^{j}_t=e^{p^{j}_t}>0$ and $w^{j,l}_t=v_t\cdot y_j^{(l)}$. Clearing the strictly positive softmax denominators turns each pairwise decision ``$\text{logit}_{f(x)}(x)>\text{logit}_{l'}(x)$'' into a strict polynomial inequality $P_{x,l'}(r,w)>0$ of degree at most $k+1$ in $n_{\mathrm{var}}:=O(kL\,qN)$ variables, with integer coefficients of bit-length $O(k\log(N+1))$. A rational-point bound first controls $(r,w)$. We then quantify the distance of this rational point from every decision boundary, approximate each log-score $p^{j}_t=\log r^{j}_t$ by a rational number within that distance, and realize the resulting rational score and projection tables by rational embeddings and weight matrices.

\begin{theorem}[Precision compactness]
\label{thm:prec-compact}

If $f$ is computable at all, then it is computable by a one-layer $k$-head attention-only transformer whose embedding table and weight matrices are rational and have bit-length
\[
  \widehat{\beta}\ \le\ O\bigl(k\log(N+1)\bigr)\cdot(k+1)^{O(n_{\mathrm{var}})}
  \ =\ 2^{\,O(n_{\mathrm{var}}\log(k+1))},\qquad n_{\mathrm{var}}=O(kL\,qN).
\]
The realization may be taken in embedding dimension at most $qN+1$.

\end{theorem}
\noindent\emph{Proof deferred to \Cref{app:deferred-proofs}.}
\par\noindent Thus, every computable function has a realization whose embedding dimension and numerical precision are both bounded in terms of the discrete parameters of the task.

\begin{remark}[Additive embedding: $qN\rightsquigarrow q+N$]
\label{rem:compact-additive}
 For functions that are additively realizable ($\Phi(b,i)=\tau_b+\pi_i$), the dimension bound sharpens to $d^{\star}\le\min\{kL,\,q+N\}+1$. The standard-parameter precision theorem above continues to apply with the general count $n_{\mathrm{var}}=O(kLqN)$; obtaining the sharper count $O(kL(q+N))$ while preserving additivity would require a separate generator-level semialgebraic argument, which we do not claim here. Indeed each cross-Gram entry then splits, $M_{a,(b,i)}=a\cdot\tau_b+a\cdot\pi_i$, so $M$ factors through only the $q$ value-vectors $\{\tau_b\}$ and the $N$ position-vectors $\{\pi_i\}$; the rank reconstruction of \Cref{thm:dim-compact}, applied to the $kL\times(q+N)$ matrix over $\{\tau_b\}\cup\{\pi_i\}$, returns the value- and position-vectors separately, so setting $\Phi'(b,i)=\tau'_b+\pi'_i$ preserves additivity for free.
\end{remark}

\section{Universal Bounds for General Binary Functions}
\label{sec:univ-bounds}

\subsection{Universal Upper Bound: the Monomial Vote}
\label{sec:univ-ub}

With enough heads a single attention layer computes an \emph{arbitrary} Boolean function, and it does so even under the restricted \emph{additive} embedding $\Phi(b,i)=\tau_b+\pi_i$, using one head per monomial of the target's multilinear expansion, in embedding dimension $O(n)$ and with $O(n)$-bit parameters. This is the construction we use to bracket the hierarchy from above. Throughout, $n$ is the input arity: $f:\{0,1\}^n\to\{0,1\}$, presented as in \Cref{def:k-attn} with an additive $\Phi$.

\paragraph{The multilinear expansion.} Every $f:\{0,1\}^n\to\{0,1\}$ agrees on the cube with its unique \emph{multilinear} (M\"obius) expansion
\[
  f(x)=\sum_{S\subseteq[n]}\alpha_S\prod_{i\in S}x_i,
  \qquad
  \alpha_S=\sum_{T\subseteq S}(-1)^{|S|-|T|}f(\mathbf 1_T)\in\mathbb Z,
\]
where $\mathbf 1_T\in\{0,1\}^n$ is the indicator of $T\subseteq[n]$. The coefficients are integers with $|\alpha_S|\le 2^{|S|}$, so $\sum_S|\alpha_S|\le\sum_S 2^{|S|}=3^{n}$. Crucially, the monomials $\prod_{i\in S}x_i$ are \emph{monotone} (positive-literal) terms --- ``every bit of $S$ is on'' --- and a monotone term is exactly the primitive a single additive head detects sharply, using the query token itself as its reference (so no auxiliary token is added to the input).

\begin{lemma}[Monotone-term head]
\label{lem:monotone-term}
For every $S\subseteq[n]$ and every $R>0$ there is a single additive head, of embedding dimension $O(n)$ and $O(R)$-bit parameters, whose scalar query read-out $A_S(x)$ satisfies
\[
  \bigl|\,A_S(x)-\textstyle\prod_{i\in S}x_i\,\bigr|\ \le\ n\,e^{-R/2}
  \qquad\text{for every }x\in\{0,1\}^n.
\]
\end{lemma}
\noindent\emph{Proof deferred to \Cref{app:deferred-proofs}.}

\begin{theorem}[Monomial vote]
\label{thm:monomial-vote}
Every $f:\{0,1\}^n\to\{0,1\}$ is computed by an additive transformer with
\[
  k\le 2^{n}\quad\text{heads},\qquad d=O(n),\qquad O(n)\text{-bit parameters}.
\]
\end{theorem}
\noindent\emph{Proof deferred to \Cref{app:deferred-proofs}.}

\subsection{A Universal Lower Bound: Almost All Functions Need \texorpdfstring{$\Omega(2^{n}/n^{2})$}{Omega(2 to the n divided by n squared)} Heads}
\label{sec:univ-lb}

The monomial vote (\Cref{thm:monomial-vote}) computes any $f:\{0,1\}^n\to\{0,1\}$ with at most $2^{n}$ heads. We now show this is essentially unavoidable: a $1-o(1)$ fraction of all Boolean functions require $\Omega(2^{n}/n^{2})$ heads, so the two bounds meet to within a $\operatorname{poly}(n)$ factor. The argument is a counting one (in the tradition of Shannon's bound that almost all Boolean functions require large circuits~\cite{shannon1949}), and, crucially, it holds for the general \emph{joint} embedding, even with \emph{unbounded} embedding dimension and \emph{unbounded} numerical precision. Neither of those two analog resources can substitute for heads.

\subsubsection{Each input's decision is one polynomial sign in \texorpdfstring{$O(kn)$}{O(kn)} variables}

Fix a $k$-head transformer with a joint embedding and use the affine scalar convention of \Cref{def:ltf}. For head $j$, put $c_j:=V_j\omega\in\mathbb R^d$ and define
\[
  r_{b,i}^j:=e^{v_?A_jv_{b,i}^\top}>0,
  \qquad u_{b,i}^j:=v_{b,i}\cdot c_j,
  \qquad r_?^j:=e^{v_?A_jv_?^\top}>0,
  \qquad u_?^j:=v_?\cdot c_j.
\]
Then
\[
  d_x^j:=r_?^j+\sum_{i=1}^nr_{x_i,i}^j>0,
  \qquad s_x^j:=r_?^ju_?^j+\sum_{i=1}^nr_{x_i,i}^ju_{x_i,i}^j,
\]
and the decision statistic is
\[
  G(x):=\sum_{j=1}^k\frac{s_x^j}{d_x^j}-\theta,
  \qquad \mathrm{net}(x)=\mathbf1\{G(x)>0\}.
\]

\begin{lemma}[Dimension- and precision-free reduction]
\label{lem:reduction}

The function computed by a $k$-head affine-scalar transformer depends on its parameters $(\Phi,\{A_j,V_j\},\omega,\theta)$ only through the at most $4kn+2k$ reals
\[
  \{r_{b,i}^j,u_{b,i}^j\}_{b\in\{0,1\},\,i\in[n],\,j\in[k]}
  \cup\{r_?^j,u_?^j\}_{j\in[k]},
\]
together with $\theta$. This list is independent of embedding dimension, value width, and any precision bound.
\end{lemma}
\begin{proof}

The displayed formulas reconstruct every $s_x^j,d_x^j$, hence $G(x)$, from the listed scalars. The triples $(b,i,j)$ contribute $2kn$ weights and $2kn$ scalar values, and the query contributes two more scalars per head. The ambient dimensions occur only inside the defining products and disappear after this reduction.
\end{proof}

Collect these numbers into a single parameter vector $y\in\mathbb{R}^{v}$, $v=O(kn)$ (at most $4kn+2k$ weights and scalar values, plus $\theta$). Because every $d^{j}_x>0$, multiplying $G(x)>0$ through by the positive product $\prod_j d^{j}_x$ is sign-preserving, so
\begin{equation}\label{eq:ulb-Px}
  \mathrm{net}(x)=\mathbf 1\{P_x(y)>0\},\qquad
  P_x(y)=\sum_{j=1}^{k}s^{j}_x\!\!\prod_{j'\neq j}\! d^{j'}_x\;-\;\theta\prod_{j=1}^{k}d^{j}_x .
\end{equation}
Each $d^{j}_x$ is \emph{linear} in $y$ and each $s^{j}_x$ is \emph{quadratic} (a sum of products $r\cdot u$), so a term $s^{j}_x\prod_{j'\neq j}d^{j'}_x$ has degree $2+(k-1)=k+1$ in $y$, and $\theta\prod_j d^{j}_x$ has degree $k+1$; hence $\deg_y P_x\le k+1$. Two points make this a genuine \emph{relaxation}, which is all the counting needs. First, we impose no constraint on $y$ beyond $r^{j}_{b,i}>0$ (in particular, we do not require the additive tie of a token-plus-position embedding), so the argument covers the joint model (and additive as a special case). Second, letting $y$ range over all of $\mathbb{R}^{v}$ can only enlarge the set of achievable decision patterns, and $v=O(kn)$ is independent of $d$ and of any precision bound.

\subsubsection{A small polynomial system has few sign patterns}

A transformer's entire input--output behavior is the vector $(\mathrm{net}(x))_{x\in\{0,1\}^n}$, read off by \eqref{eq:ulb-Px} from the signs of the $2^{n}$ polynomials $P_x$ at the single point $y$. How many such sign vectors are possible as $y$ ranges over $\mathbb{R}^{v}$? This is bounded by a classical result of Warren~\cite{warren1968}, refining the Milnor--Thom bound~\cite{milnor1964,thom1965} on the number of connected components of a real algebraic set.

\begin{theorem}[Warren's sign-pattern bound]
\label{thm:warren}
Let $P_1,\dots,P_m$ be real polynomials in $v$ real variables, each of degree at most $D$, with $m\ge v\ge 1$. The number of distinct nonzero sign vectors
\[
  \bigl(\operatorname{sgn}P_1(y),\dots,\operatorname{sgn}P_m(y)\bigr)\in\{-1,+1\}^{m},\qquad y\in\mathbb{R}^{v},
\]
that actually occur is at most $(C D m/v)^{v}$ for an absolute constant $C$. The bound depends only on $m,D,v$, never on the coefficients; in particular, it presupposes no bound on precision.
\end{theorem}

\subsubsection{Counting the realizable functions}

\begin{theorem}[Universal lower bound]
\label{thm:univ-lb}
The number of Boolean functions $f:\{0,1\}^n\to\{0,1\}$ computable by a one-layer $k$-attention-only transformer (with a joint embedding, at any dimension and any precision) is at most $2^{\,O(k n^{2})}$. Consequently a $1-o(1)$ fraction of all $2^{2^{n}}$ Boolean functions require
\[
  k\;=\;\Omega\!\Bigl(\frac{2^{n}}{n^{2}}\Bigr)\quad\text{heads},
\]
even with unbounded embedding dimension and unbounded precision.
\end{theorem}
\begin{proof}
Before applying Warren's bound, we may assume that every decision is strict. Given a realization, replace $\theta$ by $\theta+\varepsilon$, where
\[
  0<\varepsilon<\min_{x:\,f(x)=1}G(x)
\]
when $f^{-1}(1)\ne\varnothing$, and take any $\varepsilon>0$ otherwise. The finitely many positive margins remain positive, while every input previously satisfying $G(x)\le0$ becomes strictly negative. Thus the computed function is unchanged and $P_x(y)\ne0$ for every input $x$.

Apply \Cref{thm:warren} to the system $\{P_x\}_{x\in\{0,1\}^n}$ of \eqref{eq:ulb-Px}: here $m=2^{n}$, the degree is $D\le k+1$, and the number of variables is $v=O(kn)$. We split into two cases. If $m<v$, then the trivial bound by the total number of Boolean functions gives
\[
  \#\{\text{realizable }f\}\ \le\ 2^{m}\ \le\ 2^{\,O(kn^{2})},
\]
where the second inequality follows from $m<v=O(kn)$. We may therefore assume that $m\ge v$, in which case \Cref{thm:warren} applies. Since $\mathrm{net}(x)$ is determined by $\operatorname{sgn}P_x(y)$, two transformers computing different functions must realize different sign vectors, and hence
\[
  \#\{\text{realizable }f\}\ \le\ \Bigl(\tfrac{C(k+1)2^{n}}{v}\Bigr)^{v}\ =\ 2^{\,O(v n)}\ =\ 2^{\,O(k n^{2})},
\]
using $v=O(kn)$ and, since $k+1\le 2v\le 2m$, $\log\!\tfrac{C(k+1)2^{n}}{v}\le n+O(1)=O(n)$. Thus, in both cases, the number of realizable Boolean functions is at most $2^{\,O(kn^{2})}$. The bound holds for all $y\in\mathbb{R}^{v}$, hence even when precision and embedding dimension are unbounded, and (by the relaxation above) the joint embedding. Write the bound as $2^{\,c k n^{2}}$. To realize even a $2^{-o(2^{n})}$ fraction of the $2^{2^{n}}$ functions one needs $2^{\,c k n^{2}}\ge 2^{\,2^{n}/2}$, i.e.\ $k\ge 2^{n}/(2cn^{2})=\Omega(2^{n}/n^{2})$. Contrapositively, if $k<\tfrac{1}{2c}\,2^{n}/n^{2}$ then at most $2^{\,2^{n}/2}$ functions are realizable (a $2^{-\Omega(2^{n})}$ fraction of the total), so a $1-o(1)$ fraction of all $f$ need $\Omega(2^{n}/n^{2})$ heads.
\end{proof}

\begin{remark}[The bounds meet; head count is the sole resource]
The additive upper bound (\Cref{thm:monomial-vote}, $\le 2^{n}$ heads) and this joint lower bound ($\Omega(2^{n}/n^{2})$ heads) bracket the head-complexity of a generic Boolean function to within $\operatorname{poly}(n)$.
Because the upper bound is additive and the lower bound is joint, and additive $\subseteq$ joint, the two meet for \emph{both} embedding regimes at once. Neither the embedding dimension (which the constructions keep at $O(n)$) nor the precision ($O(n)$ bits) is ever the bottleneck, since \Cref{lem:reduction} and \Cref{thm:warren} make the lower bound blind to both. The one $\operatorname{poly}(n)$ gap that remains --- closing $\Omega(2^{n}/n^{2})$ up to the trivial $2^{n}$ from below --- is a genuine open problem: a single layer has no composition, so heads cannot reuse sub-computations the way a Lupanov-type circuit does.
\end{remark}

\section{Conclusion}
\label{sec:conclusion}

We have made the number of attention heads a precise complexity measure for a single self-attention layer. Three results support the claim: heads form an exact, per-head hierarchy pinned by parity (\Cref{thm:parity-lb,thm:parity-ub}); the two continuous resources, embedding dimension and  numerical precision, provably collapse to the discrete data of the task (\Cref{thm:dim-compact,thm:prec-compact}); and the head complexity of a generic $n$-bit function is $2^{n}/\operatorname{poly}(n)$ (\Cref{thm:univ-lb}). What the measure still hides are two \emph{explicit-function} questions, one within a layer and one across layers, each asking for a named function that provably outruns a bounded budget of heads. Parity is inexpensive in this model, motivating the search for explicit functions with large head complexity.

\paragraph{An explicit hard function for one layer.}
Our universal lower bound is a counting argument, in the tradition of Shannon's circuit-size bound~\cite{shannon1949}: almost every $f$ has head complexity $2^{n}/\operatorname{poly}(n)$, yet the proof produces no particular witness. It remains open to exhibit an \emph{explicit}, natural family with exponential head complexity, thereby complementing the counting lower bound with a concrete witness. A concrete first step is quantitative: our lower and upper bounds pin the worst case only to within a $\operatorname{poly}(n)$ factor ($2^{n}/n^{2}$ versus $2^{n}$), and closing that gap would fix the exact worst-case head complexity of a single layer.

\paragraph{An explicit function that forces depth.}
The hierarchy of this paper varies one resource, heads, at a fixed depth of one; the sharper question varies depth, asking for the analogue of parity across layers. Is there an explicit family $f_n$ computable by an attention network of $O(1)$ (or $O(\log n)$) layers with $\operatorname{poly}(n)$ heads per layer, yet requiring $2^{\Omega(n)}$ heads in any single layer? Single-layer attention offers no obvious mechanism for composing or reusing intermediate results, and \Cref{thm:parity-lb} gives one exact manifestation of this limitation. Such a separation would clarify when additional depth can reduce head requirements. Unlike the counting bound above, it would \emph{name} the obstruction.

The two problems are one demand in different arenas: a specific function that provably breaks a fixed head budget. We have shown the budget is real and exactly ordered; naming the functions that break it is the natural next target.

\section{Related Work}
\label{sec:related}

\paragraph{Attention heads as a resource.}
Multi-head attention was introduced to let different heads attend to different relations in parallel \cite{vaswani2017attention}, and the QK/OV view of attention-only transformers gives a useful lens for studying such heads as separate computational units \cite{elhage2021mathematical}. Mechanistic work on induction heads further suggests that individual heads implement recognizable algorithmic subroutines \cite{olsson2022induction}. Prior theoretical work has studied limitations of pure attention with depth \cite{dong2021rankcollapse}, memorization capacity as a function of the number of heads \cite{mahdavi2024memorization}, and approximation-theoretic phase transitions in head count for one-layer transformers with feed-forward blocks \cite{yu2026headcount}. Closest in spirit to our work is the two-head separation for ESP \cite{tesfaye2026twoheads}, which shows that two narrow heads can solve tasks that one arbitrarily wide head cannot. Our results extend this head-count viewpoint from the $1$-vs.-$2$ case to an exact $k$-vs.-$(k+1)$ hierarchy.

\paragraph{Parity and formal-language limitations.}
Parity is a classical obstruction for shallow models and low-order representations \cite{minsky1969perceptrons,paturi1994rational}. In transformers, parity and related formal languages have been used as canonical tests of expressivity and learnability \cite{bhattamishra2020ability,strobl2024survey}. Hahn proved length-asymptotic limitations of self-attention on parity-like languages \cite{hahn2020theoretical}, while Chiang and Cholak showed that additional depth and feed-forward computation can overcome such limitations \cite{chiang2022overcoming}. Similarly, transformers with feed-forward computation can learn counting-style shortcuts to automata, including parity-like behavior, by first aggregating counts and then post-processing them \cite{liu2023shortcuts}. These results are consistent with our setting: our hierarchy is for a self-attention layer, where the number of heads is the computational resource being measured.

The closest prior parity lower bounds are those of Kozachinskiy, Steifer, and Wa\l{}\k{e}ga \cite{kozachinskiy2026parity} and Hsu \cite{hsu2026parity}, discussed in \Cref{subsec:closest-prior}: both establish that parity is hard for one layer, but neither gives an exact head threshold. The novelty of our result is not the qualitative hardness of parity, but the exact unit-tight head hierarchy and its matching upper bound.

\paragraph{One-layer lower-bound techniques.}
Several lower bounds for one-layer transformers use communication complexity, precision restrictions, or dimension/parameter tradeoffs. Sanford, Hsu, and Telgarsky prove representational lower bounds for transformers via communication complexity \cite{sanford2023representational} and show that one-layer transformers need $hmp=\Omega(n)$ resources for induction-head tasks when head count, dimension, and precision are all counted \cite{sanford2024onelayer}. Peng, Narayanan, and Papadimitriou prove limitations for composition and graph-style reasoning under related resource assumptions \cite{peng2024limitations}. The authors in \cite{yehudai2024count} study when transformers can count and identify dimension/width requirements for extracting counts from softmax ratios. Kozachinskiy et al.~\cite{kozachinskiy2025strassen} introduce the split-VC-dimension method for infinite-precision one-layer lower bounds; they separately propose Strassen attention as an alternative attention mechanism that solves the tasks used to establish those lower bounds. These approaches are powerful for size, precision, depth, or approximation tradeoffs, but they do not yield an unconditional finite statement of the form ``$k$ heads fail, regardless of dimension and precision.''

\paragraph{Depth and parallelism.}
A complementary line of work studies depth as a computational resource. Saturated and log-precision transformers are connected to constant-depth threshold circuits \cite{merrill2022saturated,merrill2023parallelism}, logarithmic depth increases expressive power \cite{sanford2024logdepth,merrill2025depth}, and pointer-chasing and composition tasks yield further limitations for shallow transformer computation \cite{peng2024limitations}. From a learning perspective, Ye, Fu, Jia, and Sharan prove that an $L$-layer Disentangled Transformer can implement graph connectivity up to diameter $3^L$, while the training distribution determines whether gradient descent learns this algorithm or a degree-based shortcut \cite{ye2026graphconnectivity}.

\paragraph{Universality, positional encodings, and scope.}
Universality and Turing-completeness results for transformers rely on resources outside our exact hierarchy setting, such as additional layers, feed-forward computation, approximation rather than exact finite computation, or unbounded decoding time \cite{yun2020universal,kajitsuka2024onelayer,perez2021turing}. Strobl, Angluin, and Frank show that concise one-layer transformers can evaluate functions when the function table is supplied in context \cite{strobl2025concise}; our arbitrary-function construction instead encodes the function in the weights and measures the number of heads required. Finally, work on fast attention and subquadratic alternatives studies computational efficiency rather than expressivity \cite{alman2023fast,alman2025subquadratic}.

\appendix
\section{Deferred Proofs}
\label{app:deferred-proofs}

This appendix collects the proofs deferred from the main text, in order of appearance.

\subsection{Equivalence of the two read-outs}

\begin{theorem}[The arg-max and affine scalar read-outs are equivalent on Boolean tasks]
\label{thm:readout-equiv}
Fix an attention core producing $\mathrm{MHA}(x)\in\mathbb{R}^{d_{\mathrm{out}}}$.
\begin{enumerate}
\item Every two-logit arg-max read-out has an affine scalar read-out on the same core computing the same function; take $\omega=w_1-w_0$ and $\theta=0$.
\item Every affine scalar read-out has an arg-max realization computing the same function after adding one embedding coordinate and one value coordinate, but no head.
\end{enumerate}
Consequently the two read-outs compute exactly the same Boolean functions at every head count.
\end{theorem}

\begin{proof}[Proof of \Cref{thm:readout-equiv}]
\emph{Part (1).} By definition,
\[
\begin{aligned}
  \mathrm{out}(x)=1
  &\iff Z(x)_1>Z(x)_0\\
  &\iff \langle w_1,\mathrm{MHA}(x)\rangle
          >\langle w_0,\mathrm{MHA}(x)\rangle\\
  &\iff \langle w_1-w_0,\mathrm{MHA}(x)\rangle>0.
\end{aligned}
\]
Taking $\omega:=w_1-w_0$ and $\theta:=0$ gives
\[
  \mathrm{net}(x)=\mathbf 1\{\langle\omega,\mathrm{MHA}(x)\rangle-0>0\}=\mathrm{out}(x)\qquad\text{for every }x.
\]
No coordinate is added and the core is untouched, so the LTF model computes the same function.

\medskip
\emph{Part (2).} We give the context one coordinate that equals the same constant on every input, and then let the arg-max read-out use that coordinate as a bias. Mark augmented objects with a superscript $+$.

\emph{Step 1 (embedding).} Append a coordinate $d+1$ to every token embedding, set to $1$: $(v_{b,i})_{d+1}=1$ for all $b,i$ and $\Phi(?)_{d+1}=1$; leave coordinates $1,\dots,d$ unchanged. Then every token's augmented embedding has last coordinate $1$.

\emph{Step 2 (attention and values).} Extend each attention matrix to $A_j^{+}\in\mathbb{R}^{(d+1)\times(d+1)}$, equal to $A_j$ on the top-left $d\times d$ block and $0$ in the new row and column. Since the query and token embeddings agree with the originals on the first $d$ coordinates and the new row/column of $A_j^{+}$ is zero, every score is unchanged,
\[
  \Phi^{+}(?)\,A_j^{+}\,(X_i^{+})^{\top}=\Phi(?)\,A_j\,X_i^{\top},
\]
hence so are all softmax weights. Extend each value matrix to $V_j^{+}\in\mathbb{R}^{(d+1)\times(d_{\mathrm{out}}+1)}$, equal to $V_j$ on the top-left $d\times d_{\mathrm{out}}$ block, and mapping the new input coordinate $d+1$ to the new value coordinate $d_{\mathrm{out}}+1$; thus $(X_i^{+}V_j^{+})_{d_{\mathrm{out}}+1}=(X_i^{+})_{d+1}=1$ and $(X_i^{+}V_j^{+})_{1:d_{\mathrm{out}}}=X_iV_j$.

\emph{Step 3 (a constant coordinate).} For each head the new context coordinate is a softmax average of $1$'s, hence equal to $1$; summing the $k$ heads,
\[
  \mathrm{MHA}^{+}(x)_{d_{\mathrm{out}}+1}=\sum_{j=1}^{k}1=k=:C\quad\text{for every }x,
  \qquad \mathrm{MHA}^{+}(x)_{1:d_{\mathrm{out}}}=\mathrm{MHA}(x),
\]
so $\mathrm{MHA}^{+}(x)=\big(\mathrm{MHA}(x),\,C\big)$ with $C=k$ a fixed nonzero constant.

\emph{Step 4 (read-out).} Take the output matrix $W_O^{+}\in\mathbb{R}^{2\times(d_{\mathrm{out}}+1)}$ with rows $w_0^{+}=\mathbf 0$ and $w_1^{+}=\big(\omega,\,-\theta/C\big)$. Then for every $x$,
\[
  Z^{+}(x)_1-Z^{+}(x)_0=\langle w_1^{+},\mathrm{MHA}^{+}(x)\rangle
  =\langle\omega,\mathrm{MHA}(x)\rangle+\Big(-\tfrac{\theta}{C}\Big)\cdot C
  =\langle\omega,\mathrm{MHA}(x)\rangle-\theta.
\]
Hence $\mathrm{out}^{+}(x)=1\iff\langle\omega,\mathrm{MHA}(x)\rangle-\theta>0\iff\mathrm{net}(x)=1$. The simulation used exactly one extra value coordinate (and one extra embedding coordinate feeding it), completing the proof.
\end{proof}

\paragraph{Value width.}
In the arg-max form, the Boolean decision depends on $V_j$ and $W_O$ only through $V_j(w_0-w_1)^{\top}\in\mathbb{R}^d$. For any fixed $w_0\ne w_1$, the map $V_j\mapsto V_j(w_0-w_1)^{\top}$ is onto $\mathbb{R}^d$ for every $d_{\mathrm{out}}\ge1$. Thus value width does not change the class of Boolean functions computed at a fixed head count.

\begin{remark}[Where the threshold comes from]
Part~(1) always yields threshold $\theta=0$: each logit $\langle w_s,\mathrm{MHA}(x)\rangle$ is homogeneous in $\mathrm{MHA}(x)$, so the two logits are equal at $\mathrm{MHA}(x)=0$ and the arg-max boundary passes through the origin. A nonzero threshold is an affine, not a homogeneous, condition, and the constant coordinate installed in Part~(2) is exactly what lets a homogeneous read-out express it. When the core already carries such a coordinate (for instance an anchor token whose fixed value every head averages in), no augmentation is needed, and the two read-outs coincide with no overhead.
\end{remark}

\subsection{The parity upper bound}

\begin{lemma}[Alternating partial fractions]\label{lem:parity-ep-pf}
Let
\[
D(t)=\prod_{j=1}^k(t+j),
\qquad
P(t)=(-1)^k\prod_{m=0}^{k-1}\left(t-m-\frac12\right).
\]
There exist real numbers $\alpha_j,\beta_j$ such that
\[
\frac{P(t)}{D(t)}=\sum_{j=1}^k\frac{\alpha_j+\beta_jt}{t+j}.
\]
Moreover, $D(c)>0$ and $\operatorname{sign}P(c)=(-1)^c$ for every $c\in\{0,\ldots,k\}$.
\end{lemma}
\begin{proof}
Since $\deg P=\deg D$ and the roots $-1,\ldots,-k$ of $D$ are simple,
\[
\frac{P(t)}{D(t)}=C+\sum_{j=1}^k\frac{\lambda_j}{t+j}
\]
for some $\lambda_1,\ldots,\lambda_k\in\mathbb R$, where $C=(-1)^k$ is the ratio of the leading coefficients. Absorb the constant term as
\[
C=\sum_{j=1}^k \frac{C}{k}\,\frac{t+j}{t+j},
\]
and set $\alpha_j=\lambda_j+Cj/k$ and $\beta_j=C/k$. This gives the claimed form.

For the sign claim, $D(c)>0$ on $\{0,\ldots,k\}$. The roots of $P$ are half-integers, so no factor of $P(c)$ vanishes at an integer. Exactly $k-c$ factors are negative, so
\[
\operatorname{sign}P(c)=(-1)^k(-1)^{k-c}=(-1)^c.
\qedhere
\]
\end{proof}

\begin{proof}[Proof of \Cref{thm:parity-ub}]
Write $c=\sum_{i=1}^k x_i$ for the Hamming weight, whose parity sign is $(-1)^c$, with $c\in\{0,\ldots,k\}$. We construct the output gap as a rational function $\Delta(c)$ with sign $(-1)^c$ on these $k+1$ values.

\medskip
\noindent\textbf{Embedding and scores.}\quad
Set $d=k+3$ and use the standard basis $e_1,\ldots,e_{k+3}$: the first two coordinates encode the token value and the remaining $k+1$ coordinates the positions:
\[
\Phi(0,i)=e_{1}+e_{i+2}, \Phi(1,i)=e_2 + e_{i+2}
\quad(1\le i\le k),
\qquad
v_?:=\Phi(?)=e_{k+3}.
\]
Specify each $A_j$ through its effective query $\xi_j:=v_?A_j$. For $j=1,\ldots,k$, set
\begin{equation}\label{eq:parity-ep-xi}
\xi_j=\log(j/k)e_1+\log(1+j/k)e_2-Me_{k+3},
\end{equation}
with $M\to\infty$ so that the query slot is masked. Since $\xi_j$ has no positional coordinates, every bit position scores $\log(j/k)$ on token $0$ and $\log(1+j/k)$ on token $1$, and the query position scores $-M$. Hence,

ignoring the vanishing query self-weight for the moment, the scalar denominator is
\[
  d_x^j=(k-c)\frac{j}{k}+c\left(1+\frac{j}{k}\right)=c+j.
\]

\medskip
\noindent\textbf{Values.}\quad

By \Cref{lem:parity-ep-pf}, fix $\alpha_j,\beta_j$ with $P(t)/D(t)=\sum_{j=1}^k(\alpha_j+\beta_jt)/(t+j)$, and define
\[
  \eta_{0,j}:=\frac{\alpha_j}{j},\qquad
  \eta_{1,j}:=\frac{\beta_j+\alpha_j/k}{1+j/k},\qquad
  \delta_j:=\eta_{0,j}e_1+\eta_{1,j}e_2.
\]
To realize this scalar readout in the model, take $d_{\mathrm{out}}=2$ and $W_O=I_2$, and choose $V_j$ so that $V_je_1=\delta_j/2$ and $V_je_2=-\delta_j/2$. Then $V_j(w_0-w_1)^\top=\delta_j$, exactly as in \Cref{sec:prelim}. Since $\delta_j$ has no positional coordinates,
\[
  u^j_{0,i}=\Phi(0,i)\cdot\delta_j=\eta_{0,j},\qquad
  u^j_{1,i}=\Phi(1,i)\cdot\delta_j=\eta_{1,j},\qquad
  u^j_?=v_?\cdot\delta_j=0.
\]
Hence the scalar numerator is
\[
  s_x^j=(k-c)\frac{j}{k}\eta_{0,j}
  +c\left(1+\frac{j}{k}\right)\eta_{1,j}
  =\alpha_j+\beta_jc.
\]
\medskip
\noindent\textbf{Total gap.}\quad
Summing over the heads,
\[
\Delta(x)=\Delta(c)=\sum_{j=1}^k\frac{s_x^j}{d_x^j}
=\sum_{j=1}^k\frac{\alpha_j+\beta_jc}{c+j}
=\frac{P(c)}{D(c)},
\]
and by \Cref{lem:parity-ep-pf} its sign is $(-1)^c=(-1)^{\sum_{i=1}^k x_i}$ for every $c\in\{0,\ldots,k\}$, as required.

\medskip
\noindent\textbf{Finally.}\quad The computation above masked the query slot exactly ($M\to\infty$); a sufficiently large finite $M$ suffices. Indeed, the query contributes $\varepsilon=e^{-M}$ to each denominator and, since $u^j_?=v_?\cdot\delta_j=0$, nothing to any numerator gap, so the finite-$M$ gap is $\sum_{j=1}^k(\alpha_j+\beta_jc)/(c+j+\varepsilon)\to P(c)/D(c)$ as $\varepsilon\to0$. Since the limit is nonzero for each of the finitely many $c\in\{0,\ldots,k\}$, all $2^k$ strict inequalities hold once $M$ is large enough. Therefore the transformer outputs the parity bit.
\end{proof}

\subsection{Compactness}
\begin{proof}[Proof of the upper bound in \Cref{thm:dim-compact}]
Factor $M = P\Sigma Q^\top$ (thin SVD, $\Sigma \succ 0$ of size $r\times r$,
$r=\operatorname{rank}(M)$), and set
\[
a' := \Sigma^{1/2}P_{a,\cdot} \in \mathbb{R}^r \ (a\in A), \qquad
v_t' := \Sigma^{1/2}Q_{t,\cdot} \in \mathbb{R}^r \ (t\in B),
\]
so that $a'\cdot v_t' = M_{a,t}$ for all $a,t$; in particular $v_?' \in \mathbb{R}^r$ is the
reconstructed query vector, possibly $0$.

Work now in dimension $r+1$. Append one new coordinate:
\[
v_t := (v_t',\,0) \in \mathbb{R}^{r+1} \quad (t \in B,\ t\neq\, ?), \qquad
\Phi(?) := (v_?',\,1) \in \mathbb{R}^{r+1}.
\]

For each $j$, define $A_j \in \mathbb{R}^{(r+1)\times(r+1)}$ by: its top $r$ rows are $0$, and its
last row is $(\xi_j', 0)$, where $\xi_j' := a'$ is the reconstructed vector for $a=\xi_j\in A$. Then,
writing $\Phi(?) = (v_?', 1)$,
\[
\Phi(?)A_j \;=\; v_?'\cdot(\text{top }r\text{ rows of }A_j) \;+\; 1\cdot(\xi_j', 0)
\;=\; (\xi_j', 0),
\]
regardless of $v_?'$ — the top rows being zero exactly cancels any dependence on the old query
coordinates. So set $\xi_j := (\xi_j',0) \in \mathbb{R}^{r+1}$; we have engineered $\Phi(?)A_j = \xi_j$
exactly, as the model requires.

It remains to check $\xi_j$ still reproduces every entry of $M$ it is responsible for. For an
input token $t\neq\, ?$: $\xi_j \cdot v_t = (\xi_j',0)\cdot(v_t',0) = \xi_j'\cdot v_t' = M_{\xi_j,t}$,
unchanged. For the self-score, $t=\,?$: $\xi_j \cdot \Phi(?) = (\xi_j',0)\cdot(v_?',1) = \xi_j'\cdot
v_?' = M_{\xi_j,?}$.

The remaining objects — $\Phi(?)$'s use as a reference point for value-projections
$w_t^{j,l} = v_t\cdot y_j^{(l)}$, and the read-out vectors $y_j^{(l)}$ themselves — are reconstructed
exactly as before (padded with a trailing $0$), so all other entries of $M$ are reproduced
unchanged. The reassembled transformer therefore has dimension $r+1$ and computes the same
function.
\end{proof}

\begin{proof}[Proof of \Cref{thm:prec-compact}]
Under the tie-breaking rule a realizer may satisfy some pairwise comparisons with equality, so we first make every comparison strict. Fix a realizer and, for each label $l$, shift head $1$'s value-projections to label $l$ by a constant, $w^{1,l}_t\to w^{1,l}_t+\delta_l$, for \emph{every} occurring token type $t$, the query token included. Because the shift is uniform over tokens, it adds $\delta_l\,d^{1}_x/d^{1}_x=\delta_l$ to $\mathrm{logit}_l(x)$ on every input $x$. Ordering the $\delta_l$ by the tie-breaking priority (the preferred label receiving the larger shift) and taking them small enough --- there are finitely many inputs and labels --- turns every tie into a strict inequality in favor of its original winner while preserving every previously strict comparison.
Hence the strict system consisting of the sign conditions $\{\sigma_{x,l'}P_{x,l'}>0\}$ over the $\le q^{N}(L-1)$ pairs $(x,l')$, \emph{together with} the positivity constraints $\{r^{j}_t>0\}$ (degree-$1$ inequalities, required so that any solution represents exponentiated attention weights), is nonempty --- the shifted realizer satisfies all of them --- and its solution set is open, hence full-dimensional and basic semialgebraic; any of its points computes $f$ under the tie-breaking semantics.
By the standard point-bit-length bounds for such sets~\cite{renegar1992,basu2006algorithms}, a nonempty full-dimensional set cut out by degree-$\le\delta$ polynomials in $n_{\mathrm{var}}$ variables with $\le\tau$-bit integer coefficients contains a rational point of bit-length $\le\tau\,\delta^{O(n_{\mathrm{var}})}$. Here $\delta=k+1$ and $\tau=O(k\log(N+1))$, giving $\beta\le O\bigl(k\log(N+1)\bigr)(k+1)^{O(n_{\mathrm{var}})}=2^{O(n_{\mathrm{var}}\log(k+1))}$.

Let $z^\star=(r^\star,w^\star)\in\mathbb{Q}^{n_{\mathrm{var}}}$ be this $\beta$-bit rational point. We first quantify how accurately its log-scores must be rounded. For any strict signed decision polynomial $Q=\sigma_{x,l'}P_{x,l'}$, clearing the coordinate denominators in $Q(z^\star)>0$ shows that
\[
  Q(z^\star)\ \ge\ 2^{-\beta\delta n_{\mathrm{var}}}.
\]
Indeed, a common denominator divides the product of the $n_{\mathrm{var}}$ coordinate denominators raised to the degree $\delta$, and the resulting numerator is a positive integer. On the unit $\ell_\infty$-box around $z^\star$, every coordinate has magnitude at most $R:=2^{\beta+1}$. Since $Q$ has degree at most $\delta$, coefficient magnitude at most $2^\tau$, and at most $\binom{n_{\mathrm{var}}+\delta}{\delta}$ monomials, its $\ell_1$-gradient is bounded there by
\[
  \Lambda
  \ :=\
  n_{\mathrm{var}}\delta\binom{n_{\mathrm{var}}+\delta}{\delta}2^\tau R^{\delta-1}.
\]
The mean-value theorem therefore shows that every signed decision polynomial remains positive throughout the $\ell_\infty$-ball of radius
\[
  \varepsilon
  \ :=\
  \min\!\left\{1,\frac{2^{-\beta\delta n_{\mathrm{var}}-1}}{\Lambda}\right\}
\]
around $z^\star$. In particular,
\[
  \log_2(\varepsilon^{-1})
  \ =\
  O\!\left(\beta\delta n_{\mathrm{var}}+\tau+\delta\log(n_{\mathrm{var}}+\delta)\right).
\]

For each $r^{j,\star}_t$, put $p^{j,\star}_t=\log r^{j,\star}_t$ and choose a dyadic rational $\widehat p^j_t$ satisfying
\[
  \left|\widehat p^j_t-p^{j,\star}_t\right|
  \ \le\
  2^{-\beta-2}\varepsilon .
\]
Because a positive $\beta$-bit rational lies in $[2^{-\beta},2^\beta]$, we have $|p^{j,\star}_t|\le\beta\log 2$. The mean-value theorem for the exponential gives
\[
  \left|e^{\widehat p^j_t}-r^{j,\star}_t\right|
  \ \le\ \varepsilon.
\]
Moreover, the dyadic approximants can be chosen with bit-length
\[
  \widehat\beta
  \ =\
  O\!\left(\beta\delta n_{\mathrm{var}}+\tau+\delta\log(n_{\mathrm{var}}+\delta)\right)
  \ =\
  O\bigl(k\log(N+1)\bigr)(k+1)^{O(n_{\mathrm{var}})}
  \ =\
  2^{O(n_{\mathrm{var}}\log(k+1))}.
\]
Thus $(e^{\widehat p},w^\star)$ lies in the same realizing cell as $(r^\star,w^\star)$ and computes $f$ with rational score and value-projection tables of the claimed bit-length.

Finally, realize these tables by rational model parameters. Enumerate the $T:=|B|\le qN$ occurring token types, including the query, and work in dimension $T+1$. Assign each token type $t$ its standard basis vector $v_t=e_t$. For each head $j$, put the entries $\widehat p^j_t$ in the query row of $A_j$, so that $\Phi(?)A_jv_t^\top=\widehat p^j_t$. For each label $l$, take $y_j^{(l)}$ to have coordinates $(w_t^{j,l})^\star$ on these basis vectors and zero in the extra coordinate. Taking $d_{\mathrm{out}}=L$, $W_O=I_L$, and $V_j=Y_j$ realizes these read-out vectors exactly. All entries of $\Phi,\{A_j,V_j\},W_O$ are rational of bit-length at most $\widehat\beta$, and the resulting transformer has dimension at most $qN+1$ and computes $f$.

\end{proof}

\subsection{The additive monomial vote}

\begin{proof}[Proof of \Cref{lem:monotone-term}]
Use the query token itself as the head's reference: by \Cref{def:k-attn} the query attends to itself, so set its self-score to give $r_?=e^{-R/2}$, read the query's value as $1$, and read every bit token's value as $0$ (additively realizable: the value-projection's position part is free, so put weight $1$ on the query position and $0$ on the bit positions, with bit-slope $0$). Choose the effective query so that the exponentiated bit scores are
\[
  r_{x_i,i}=e^{-R\,x_i}\ \ (i\in S),\qquad
  r_{x_i,i}=e^{-2R-R\,x_i}\le e^{-2R}\ \ (i\notin S).
\]
This is additive: the bit-slope $\Delta=\xi\cdot(\tau_1-\tau_0)=-R$ is a single head constant (position-independent), while the position biases $h(i)=\xi\cdot\pi_i$, free in $i$ through the $\pi_i$, are $0$ on $S$ and $-2R$ off $S$ (so off-$S$ tokens score at most $e^{-2R}$ regardless of their bit). Because only the query carries value $1$, the read-out is the query's self-attention weight,
\[
  A_S(x)=\frac{e^{-R/2}}{\,e^{-R/2}+\sum_{i\in S}e^{-R\,x_i}+\sum_{i\notin S}e^{-2R-R\,x_i}\,}.
\]
If $x_i=1$ for all $i\in S$ (the term holds), the middle sum is $|S|\,e^{-R}$ and $A_S(x)=1-O(n\,e^{-R/2})$. If some $i\in S$ has $x_i=0$ (the term fails), that token contributes $e^{0}=1$ to the denominator, so $A_S(x)\le e^{-R/2}$. Either way $A_S(x)$ is within $n\,e^{-R/2}$ of $\prod_{i\in S}x_i$. The construction adds no token to the input; it uses $d=O(n)$, and all scores are $O(R)$.
\end{proof}

\begin{proof}[Proof of \Cref{thm:monomial-vote}]
Take one monotone-term head (\Cref{lem:monotone-term}) per nonempty $S$ with $\alpha_S\ne0$, all sharing a common steepness $R$, and read them out in the affine scalar form of \Cref{def:ltf} with head weights $\alpha_S$ and threshold $\theta=\tfrac12-\alpha_\varnothing$ (the empty term $\prod_{i\in\varnothing}x_i\equiv1$ is the constant $\alpha_\varnothing$, folded into $\theta$). If no such head is needed, use one zero-valued head; since $n\ge1$, the head count remains at most $2^n$. The decision statistic is then
\[
  F(x)\;=\;\Bigl(\alpha_\varnothing+\sum_{S\ne\varnothing}\alpha_S\,A_S(x)\Bigr)-\tfrac12
        \;=\;\Bigl(\underbrace{\textstyle\sum_{S}\alpha_S\!\prod_{i\in S}x_i}_{=\,f(x)}-\tfrac12\Bigr)
        \;+\;\sum_{S\ne\varnothing}\alpha_S\bigl(A_S(x)-\textstyle\prod_{i\in S}x_i\bigr).
\]
The error term is at most $n\,e^{-R/2}\sum_S|\alpha_S|\le n\,3^{n}\,e^{-R/2}$ (\Cref{lem:monotone-term}), so choosing $R=O(n)$ (concretely any $R> 2\bigl(n\ln 3+\ln(2n)\bigr)$) makes it strictly less than $\tfrac12$. Then $\operatorname{sign}F(x)=\operatorname{sign}\bigl(f(x)-\tfrac12\bigr)$, so $\mathrm{net}(x)=\mathbf 1\{F(x)>0\}=f(x)$ for all $x$. The read-out weights $\alpha_S$ are $n$-bit integers and the scores are $O(R)=O(n)$, so every parameter is $O(n)$-bit.
\end{proof}

\section{The \texorpdfstring{$k$}{k}-ESP Problem}
\label{sec:k-esp-lb}

\subsection{Problem Definition}

This task is the natural $k$-bit extension of the Endpoint Selection Problem (ESP) of~\cite{tesfaye2026twoheads}: ESP uses one selector bit, whereas $k$-ESP selects between the two endpoints using the parity of $k$ selector bits.

Fix a finite endpoint vocabulary $\mathcal{U}$ with $|\mathcal{U}|\ge 2$, disjoint from the selector symbols $\{0,1\}$ and the query symbol $?$. The $k$-ESP problem takes the $(k{+}3)$-token input
\[
x_1x_2\cdots x_k\,p\,q\,?,
\qquad x_1,\ldots,x_k\in\{0,1\},\quad p,q\in\mathcal{U},\quad p\ne q,
\]
and outputs $p$ if $\bigoplus_{i=1}^k x_i=0$ and $q$ otherwise.

\begin{remark}[Relation to parity]
Fix distinct $u_0,u_1\in\mathcal{U}$ and restrict to
\[
(p,q)=\bigl(u_{x_{k+1}},u_{1-x_{k+1}}\bigr).
\]
The required output is $u_{\oplus_{i=1}^{k+1}x_i}$; hence, under the relabeling
$u_b\leftrightarrow b$, this restriction is exactly $(k+1)$-bit parity.
\end{remark}

\subsection{Lower Bound for \texorpdfstring{$k$}{k}-ESP}

For the lower bound, fix distinct $u_0,u_1\in\mathcal{U}$, naming them so that $u_0$ has priority over $u_1$ under the model's fixed tie-break, and restrict to
$(p,q)\in\{(u_0,u_1),(u_1,u_0)\}$.  Introduce a bit $x_{k+1}$ by
$p=u_{x_{k+1}}$, $q=u_{1-x_{k+1}}$, and write $x_{k+2}=1-x_{k+1}$.
On this restriction the input is
\[
x_1x_2\cdots x_k\,u_{x_{k+1}}\,u_{x_{k+2}}\,?,
\qquad x_{k+2}=1-x_{k+1},
\]
and the target is $u_x$ with $x=\bigoplus_{1\le i\le k+1}x_i$.
Thus any transformer solving arbitrary-vocabulary $k$-ESP must solve this restricted family.

For $b\in\{0,1\}$, use the binary index $b$ to parametrize the
two-symbol restriction by writing
\[
v_{b,i}\coloneqq
\begin{cases}
\Phi(b,i), & 1\le i\le k,\\
\Phi(u_b,i), & i\in\{k+1,k+2\},
\end{cases}
\qquad
r^j_{b,i}\coloneqq
\exp\!\bigl(v_?A_jv_{b,i}^{\top}\bigr)
,\qquad v_?\coloneqq\Phi(?).
\]

Specializing the criterion \eqref{eq:prelim-computes} of \Cref{sec:prelim}, now with the per-head sums running over all $k{+}2$ token positions and the query's self-term, head $j$ has numerator $n_x^j=r^j_?\,v_?+\sum_{1\le i\le k+2} r^j_{x_i,i}\, v_{x_i,i}$ and denominator $d_x^j=r^j_?+\sum_{1\le i\le k+2} r^j_{x_i,i}>0$, where $r^j_?=e^{v_?A_jv_?^\top}>0$ is the query's (input-independent) self-weight and $a_j,b_j$ are the columns of $Y_j=V_jW_O^\top$ reading off $u_0,u_1$. Here $Y_j$ may contain columns for every output symbol in $\mathcal{U}$; only the two columns for $u_0,u_1$ enter the necessary pairwise comparison, through $a_j-b_j$.

Any transformer solving arbitrary-vocabulary $k$-ESP must, on this restriction, satisfy for every $x=(x_1,\dots,x_{k+1},x_{k+2})$ with
$x_{k+2}=1-x_{k+1}$,

\begin{equation}
(-1)^{\sum_{1\le i\le k+1} x_i} \cdot \sum_{j=1}^{k} \frac{n^j_x}{d^j_x}\cdot(a_j-b_j)
\;\begin{cases}
> 0 & \text{if } \sum_{1\le i\le k+1} x_i \text{ is odd (target } u_1\text{)},\\[2pt]
\ge 0 & \text{if } \sum_{1\le i\le k+1} x_i \text{ is even (target } u_0\text{)}.
\end{cases}
\end{equation}
(As with parity, a tie between the two output logits resolves toward $u_0$ per \Cref{def:k-attn}'s tie-break, so a $0$-target input only needs the displayed gap to be nonnegative, not positive.) The one structural feature the proof uses is that positions $k{+}1$ and $k{+}2$ both hold functions of the \emph{single} bit $x_{k+1}$.

\begin{theorem}
\label{thm:k-esp-lb}

For every finite $\mathcal{U}$ with $|\mathcal{U}|\ge2$, no one-layer $k$-head attention-only transformer solves $k$-ESP over $\mathcal{U}$.  Indeed, after fixing distinct $u_0,u_1\in\mathcal{U}$ and restricting to the two ordered pairs above, the following stronger infeasibility statement holds.

There is no choice of vectors $v_?,v_{b,i}, a_j, b_j$, and positive reals $r^j_{b,i}, r^j_?$, for
$b\in\{0,1\}$, $1\le i\le k+2$, $1\le j\le k$, such that for every
$x=(x_1,\dots,x_{k+1},x_{k+2})$ with $x_{k+2}=1-x_{k+1}$, the following holds:
\[
(-1)^{\sum_{1\le i\le k+1} x_i} \cdot \sum_{j=1}^{k} \frac{n^j_x}{d^j_x}\cdot(a_j-b_j)
\;\begin{cases}
> 0 & \text{if } \sum_{1\le i \le k+1} x_i \text{ is odd},\\[2pt]
\ge 0 & \text{if } \sum_{1\le i \le k+1} x_i \text{ is even}.
\end{cases}
\]
\end{theorem}

\begin{proof}
Write $\sigma_x = (-1)^{\sum_{1 \le i \le k+1} x_i} \in \{\pm 1\}$ for the required sign at $x$, and let $\cdot$ denote the Euclidean inner product on $\mathbb{R}^d$.

\medskip

\noindent\textbf{Step 1: Clear the denominators.}\quad Set
\[
s_x^j:=n_x^j\cdot(a_j-b_j)
=r_?^j\bigl(v_?\cdot(a_j-b_j)\bigr)
+\sum_{i=1}^{k+2}r^j_{x_i,i}\bigl(v_{x_i,i}\cdot(a_j-b_j)\bigr).
\]
In both $s_x^j$ and $d_x^j$, the query term is constant, positions $1,\ldots,k$ depend respectively on $x_1,\ldots,x_k$, and positions $k+1,k+2$ both depend only on $x_{k+1}$. Thus every summand depends on at most one of the $k+1$ free bits. Define $\Delta_{01}(x):=\sum_{j=1}^k s_x^j/d_x^j$.

The inequality at $x$ reads $\sigma_x\Delta_{01}(x)\ge0$, with strict inequality whenever $\sigma_x=-1$.
Since each $d^j_x > 0$, multiplying through by the strictly positive quantity $\prod_{j=1}^{k} d^j_x$ preserves sign:

\[
\begin{aligned}
\mathrm{sign}(\Delta_{01}(x))&=\mathrm{sign}(F_x),\\
F_x&:=\prod_{j=1}^kd_x^j\,\Delta_{01}(x)
=\sum_{j=1}^ks_x^j\prod_{j'\ne j}d_x^{j'}.
\end{aligned}
\]
The theorem is therefore equivalent to the assertion: no choice of parameters makes $\sigma_x F_x \ge 0$ hold for every $x \in \{0,1\}^{k+1}$, with strict inequality whenever $\sigma_x = -1$.

\noindent\textbf{Step 2: Alternating-sum identity.}
\begin{claim}
For every choice of the parameters,
$\sum_{x \in \{0,1\}^{k+1}} \sigma_x F_x=0$.
\end{claim}

\begin{proof}[Proof of claim]
Expand $F_x$ as a sum of monomials.  Each term in the defining sum has the form $s^j_x \prod_{j' \ne j} d^{j'}_x$, and we further expand each factor using

\[
\begin{aligned}
s_x^j&=r_?^j\bigl(v_?\cdot(a_j-b_j)\bigr)
+\sum_{i=1}^{k+2}r^j_{x_i,i}\bigl(v_{x_i,i}\cdot(a_j-b_j)\bigr),\\
d_x^{j'}&=r_?^{j'}+\sum_{i=1}^{k+2}r^{j'}_{x_i,i}.
\end{aligned}
\]
Letting $x_{k+2}=1-x_{k+1}$, a typical monomial $m$ is a product of exactly $k$ factors---one drawn from $s_x^j$ and one from each of the $k-1$ remaining denominators. Each factor depends on at most one free bit, so $m$ depends on at most $k$ of the $k+1$ coordinates. Therefore some $i^*\in\{1,\ldots,k+1\}$ is absent. Splitting the sum over $x$ by isolating $x_{i^*}$,
\[
\sum_{x \in \{0,1\}^{k+1}} \sigma_x \cdot m \;=\; \Bigg(\sum_{\substack{x_l \in \{0,1\} \\ l \ne i^*}} (-1)^{\sum_{l \ne i^*} x_l}\, m\Bigg) \cdot \underbrace{\sum_{x_{i^*} \in \{0,1\}} (-1)^{x_{i^*}}}_{= 0} \;=\; 0.
\]
Since every monomial of $F_x$ vanishes individually under the alternating sum, so does the total: $\sum_x \sigma_x F_x = 0$.
\end{proof}

\noindent\textbf{Step 3: Contradiction.}\quad
Suppose, toward a contradiction, that the required inequalities hold simultaneously: for every
$x\in\{0,1\}^{k+1}$ (recall $x_{k+2}=1-x_{k+1}$ is determined, so the free coordinates are
$x_1,\dots,x_{k+1}$),
\[
\sigma_x F_x \;\ge\; 0,
\]
with strict inequality whenever $\sigma_x=-1$ (odd weight). Since the free coordinates range over
all of $\{0,1\}^{k+1}$, exactly $2^k>0$ of them have odd weight, so at least one term is strictly
positive while the rest are nonnegative by hypothesis. Hence
\[
\sum_{x\in\{0,1\}^{k+1}} \sigma_x F_x \;>\; 0,
\]
contradicting the identity $\sum_x \sigma_x F_x = 0$ of Step~2. Therefore no choice of parameters
satisfies all $2^{k+1}$ inequalities, completing the proof.

Since any solver for arbitrary-vocabulary $k$-ESP would also solve the fixed two-symbol restriction, the same lower bound holds for every finite endpoint vocabulary $\mathcal{U}$ with $|\mathcal{U}|\ge2$.

\end{proof}

\subsection{Upper Bound for \texorpdfstring{$k$}{k}-ESP}
\label{sec:k-esp-ub}

Fix a finite endpoint vocabulary $\mathcal{U}$ with $m\coloneqq|\mathcal{U}|\ge2$, treated as disjoint from the bit-role tokens $\{0,1,?\}$.
The input for $k$-ESP is the $(k{+}3)$-token string
\[
x_1x_2\cdots x_k\,p\,q\,?,
\qquad x_i\in\{0,1\},\quad p,q\in\mathcal{U},\quad p\ne q,
\]
and the target is $p$ if $\bigoplus_{i=1}^k x_i=0$ and $q$ otherwise.
For a one-layer $H$-head transformer, fold $V_jW_O^\top$ into the
$d\times m$ matrix $Y_j$ with columns $y_j^{(u)}$, $u\in\mathcal{U}$.
On an admissible input, write $v_?=\Phi(?,k+3)$ and
$r_?^j=e^{v_?A_jv_?^\top}>0$. Head $j$ has
\[
\begin{aligned}
n_x^j&=r_?^jv_?+\sum_{i=1}^k r^j_{x_i,i}v_{x_i,i}
      +r^j_{p,k+1}v_{p,k+1}+r^j_{q,k+2}v_{q,k+2},\\
d_x^j&=r_?^j+\sum_{i=1}^k r^j_{x_i,i}+r^j_{p,k+1}+r^j_{q,k+2}>0.
\end{aligned}
\]
Thus, for the supplied pair $p,q$,
\[
Z_p(x)-Z_q(x)
=\sum_j\frac{n_x^j}{d_x^j}\cdot\bigl(y_j^{(p)}-y_j^{(q)}\bigr).
\]
The construction first takes the idealized limit $r_?^j\to0$ and then restores a sufficiently small positive self-weight.

It suffices to make the $p$-versus-$q$ logit gap have sign
$(-1)^{\sum_{i=1}^k x_i}$ while keeping every noncandidate logit below the winning candidate; the construction below makes all noncandidate logits zero and both candidate logits positive.

The lower bound shows $H=k$ heads do not suffice: restricting to two vocabulary symbols and both endpoint orders exposes $k+1$ binary coordinates, so the alternating-sum obstruction forces $H\ge k+1$.  We show this minimum is achieved.  In particular, we present embedding $\Phi$, matrix $W_O$, and attention and value matrices $A_j$'s and $V_j$'s such that $k$-ESP is solvable using $k+1$ heads.

\begin{theorem}
For every finite $\mathcal{U}$ with $|\mathcal{U}|\ge2$ and every $k\ge1$, there is an additive one-layer $(k{+}1)$-head attention-only transformer that solves $k$-ESP over $\mathcal{U}$, in embedding dimension $d=k+|\mathcal{U}|+5$.
\end{theorem}

\medskip

\noindent\textbf{Factorization.}\quad Write
\[
c\coloneqq\#\{i:1\le i\le k,\ x_i=1\},
\qquad \sigma_x\coloneqq(-1)^c.
\]
The first $k$ bits determine the required sign by rational interpolation in $c$.
The two endpoint positions play a separate role: their values copy the identities of the supplied candidates $p$ and $q$ into their own output logits, while their unequal attention masses create the per-head $p$-versus-$q$ residue.  Because those masses depend only on the two fixed positions, they add input-independent mass to each denominator and do not interfere with the interpolation in $c$.

\medskip

\noindent\textbf{Embedding, attention, and value matrices.}
Fix a bijection $\iota:\mathcal{U}\to[m]$ and write
$g_u\coloneqq e_{2+\iota(u)}$ for the vocabulary-identity coordinate of $u$.
Set $d=k+m+5$.  The first two coordinates encode the bit tokens, the next $m$ coordinates encode endpoint identities, and the final $k+3$ coordinates encode positions:
\[
\Phi_t(b)=e_{b+1}\quad(b\in\{0,1\}),
\qquad
\Phi_t(u)=g_u\quad(u\in\mathcal{U}),
\qquad
\Phi_t(?)=0,
\]
\[
\Phi_p(i)=e_{m+2+i}\quad(1\le i\le k+3),
\qquad
\Phi(t,i)=\Phi_t(t)+\Phi_p(i).
\]

With $v_?=\Phi(?,k+3)$ as above, specify each $A_j$ through its effective query $\xi_j:=v_?A_j\in\mathbb R^d$.

\begin{equation}\label{eq:kesp-xi}
\xi_j
= a_0^j e_1+a_1^j e_2
+\rho_{k+1}^j e_{m+k+3}
+\rho_{k+2}^j e_{m+k+4}
-Me_{m+k+5}.
\end{equation}

with per-head constants $a^{j}_0, a^{j}_1$, $\rho^j_{k+1}$, and $\rho^j_{k+2}$ specified below and $M \rightarrow \infty$ so that the query slot is masked.  The scores $\rho^j_{k+1}, \rho^j_{k+2}$ are indexed by the two paired \emph{positions}: they multiply the positional coordinates

$e_{m+k+3}, e_{m+k+4}$,
so each applies to its position regardless of which token occupies it.

Note that $\xi_j$ can be implemented by setting row $m+k+5$ of $A_j$ in columns $1,2,m+k+3,m+k+4,m+k+5$ to $a_0^j,a_1^j,\rho_{k+1}^j,\rho_{k+2}^j,-M$, respectively, and setting all other entries to $0$.

\paragraph{Per-head scores.}  Using \eqref{eq:kesp-xi}, we obtain the following scores for each input element:

\begin{align*}
\text{pos }i\in\{1,\ldots,k\}\ (\text{bit }x_i):\quad
&\xi_j\cdot\bigl(e_{x_i+1}+e_{m+2+i}\bigr)=a_{x_i}^j,\\
\text{pos }k+1\ (p):\quad
&\xi_j\cdot\bigl(g_p+e_{m+k+3}\bigr)=\rho_{k+1}^j,\\
\text{pos }k+2\ (q):\quad
&\xi_j\cdot\bigl(g_q+e_{m+k+4}\bigr)=\rho_{k+2}^j,\\
\text{pos }k+3\ (?):\quad&-M.
\end{align*}
Thus $r_{x_i,i}^j=e^{a_{x_i}^j}$ for $1\le i\le k$,
$r_{p,k+1}^j=e^{\rho_{k+1}^j}$, and
$r_{q,k+2}^j=e^{\rho_{k+2}^j}$, independently of the identities of $p$ and $q$.

\medskip

\noindent\textbf{The value and output matrices.}
Set $V_j=I_d$ for every head and let the row of $W_O$ corresponding to
$u\in\mathcal{U}$ be $g_u^\top$.  Equivalently, the folded output column is
$y_j^{(u)}=g_u$.  For the supplied pair $p,q$, the relevant gap direction is
\[
y_j^{(p)}-y_j^{(q)}=g_p-g_q.
\]
Every bit and position coordinate is orthogonal to $g_p-g_q$, while
\[
g_p\cdot(g_p-g_q)=1,
\qquad
g_q\cdot(g_p-g_q)=-1.
\]
Hence the first endpoint position contributes positively to the $p$-versus-$q$ gap and the second contributes negatively.  Moreover, every $u\notin\{p,q\}$ receives logit $0$, while the $p$- and $q$-logits are sums of positive attention weights and are therefore positive.

\medskip
\noindent\textbf{The gap as a rational function of $c$.}\quad

In the idealized limit $r_?^j=0$, set
\[
s_{x,p,q}^j
\coloneqq n_x^j\cdot(g_p-g_q).
\]
Then the $p$-versus-$q$ output gap is
$Z_p(x)-Z_q(x)=\sum_j s_{x,p,q}^j/d_x^j$.

Then, summing over the first $k$ ``free'' bits (of which say $c$ are $1$ and $k-c$ are $0$),
\[
d^j_x \;=\; \underbrace{(k-c)\,e^{a^j_0} + c\,e^{a^j_1}}_{\text{free bits}} \;+\; \underbrace{e^{\rho^j_{k+1}} + e^{\rho^j_{k+2}}}_{\text{paired}} \;=\; P_j + Q_j\, c,
\]
\[
P_j = k\,e^{a^j_0} + e^{\rho^j_{k+1}} + e^{\rho^j_{k+2}}, \qquad Q_j = e^{a^j_1} - e^{a^j_0},
\]
a linear function of $c$.

The first $k$ bits contribute nothing to $s_{x,p,q}^j$, and the two endpoint positions give
\[
s_{x,p,q}^j
=e^{\rho_{k+1}^j}-e^{\rho_{k+2}^j}
\eqqcolon R_j.
\]
Therefore
\[
\frac{s_{x,p,q}^j}{d_x^j}
=\frac{R_j}{P_j+Q_jc},
\]
and
\[
Z_p(x)-Z_q(x)
=\sum_{j=1}^{k+1}\frac{R_j}{P_j+Q_jc}
\eqqcolon\Delta(c).
\]
Below we choose $R_j,P_j,Q_j$ so that
$\operatorname{sign}\Delta(c)=(-1)^c$ for every
$c\in\{0,1,\ldots,k\}$; then $p$ wins for even parity and $q$ wins for odd parity.

\medskip

\noindent\textbf{Realizability and interpolation.}\quad Two lemmas reduce the design of $\Delta$ to a Cauchy interpolation.

\begin{lemma}[Realizability]\label{lem:real}
For any reals $R$, $P > |R|$, and $Q > 0$, there exist reals $a_0, a_1, \rho_{k+1}$, and $\rho_{k+2}$ realizing $(R, P, Q)$ for one head.
\end{lemma}
\begin{proof}
Indeed, set $S=\tfrac12(P+|R|)\in(|R|,P)$, $e^{\rho_{k+1}}=\tfrac12(S+R)$, $e^{\rho_{k+2}}=\tfrac12(S-R)$, $e^{a_0}=(P-S)/k$,
$e^{a_1}=Q+e^{a_0}$; all quantities are positive, and the defining equations hold.
\end{proof}
\begin{lemma}[Cauchy interpolation]\label{lem:cauchy}
For distinct reals $p_1, \ldots, p_n$ and distinct reals $q_0, \ldots, q_{n-1}$ with $q_i \ne p_j$, for any target $y \in \mathbb{R}^n$, there exists a unique $\lambda$ such that for all $0 \le i < n$, $\sum_j \lambda_j /(q_i - p_j) = y_i$.
\end{lemma}
\begin{proof}
The $n \times n$ matrix $[\,1/(q_i - p_j)\,]$ is a Cauchy matrix and is hence invertible.  Therefore, the system of equations $\sum_j \lambda_j /(q_i - p_j) = y_i$ over variables $\lambda_j$, $1 \le j \le n$, has a unique solution.
\end{proof}
\medskip
We are now ready to complete the proof of the theorem.  Set $p_j = -j$ ($j = 1, \ldots, k+1$, strictly negative) and the $q_c = c$, $c = 0, 1, \ldots, k$ (nonnegative, hence $q_c \neq p_j$).  By Lemma~\ref{lem:cauchy} (with $n = k+1$) there is a unique $\widetilde\lambda \ne 0$ with $\sum_j \widetilde\lambda_j / (c + j) = (-1)^c$ at every node.  Put
\[
\varepsilon^\ast = \min_{\substack{1 \le j \le k+1 \\ \widetilde\lambda_j \ne 0}} \frac{j}{|\widetilde\lambda_j|} > 0,
\]
well defined since $\widetilde\lambda \ne 0$; coordinates with $\widetilde\lambda_j = 0$ simply get $R_j = 0$, which poses no realizability problem below.
Fix $\varepsilon \in (0, \varepsilon^\ast)$ and for each head set $Q_j = 1$, $P_j = j$, $R_j = \varepsilon \widetilde\lambda_j$.  Then $P_j = j > \varepsilon|\widetilde\lambda_j| = |R_j|$ and $Q_j > 0$, so Lemma~\ref{lem:real} supplies head parameters $a_0^j, a_1^j, \rho^j_{k+1}, \rho^j_{k+2}$ realizing $(R_j, P_j, Q_j)$.  Since $p_j = -j = -P_j / Q_j$, the denominator of the $j$th summand of $\Delta(c)$ is $P_j + Q_j c = c - p_j$, hence
\[
\Delta(c) = \sum_{j=1}^{k+1} \frac{R_j}{c - p_j} = \varepsilon \sum_{j=1}^{k+1} \frac{\widetilde\lambda_j}{c - p_j} = \varepsilon\,(-1)^c .
\]

Thus $\operatorname{sign}\Delta(c)=(-1)^c$ on $c\in\{0,\ldots,k\}$, so $Z_p>Z_q$ when the parity is even and $Z_q>Z_p$ when it is odd.  Since all noncandidate logits are $0$ and both candidate logits are positive, the argmax is the required member of $\{p,q\}$.  This holds uniformly for every distinct $p,q\in\mathcal{U}$, proving the theorem. \hfill$\qed$

\begin{remark}
The head count is tight.  The matching lower bound rules out $k$ heads for $k$-ESP, and the construction attains $k+1$.

The vocabulary-identity channel lets the same fixed parameters copy arbitrary supplied endpoints $p,q$ into their own logits; it changes only the embedding dimension, not the interpolation.  The $k$-bit weight $c$ still consumes all $k+1$ heads through the degree count:

over the common denominator $\prod_j (P_j + Q_j c)$, the numerator $\sum_j R_j \prod_{j' \ne j}(P_{j'} + Q_{j'} c)$ has degree $\le k$, and realizing $k$ interior sign changes of $\Delta$ across $\{0, \ldots, k\}$ requires $k$ real roots, forcing $H - 1 \ge k$.

Finally, the query slot need not be masked with an infinite score.
For finite $M$, let $\eta=e^{-M}$.

Since the query embedding has zero projection onto $g_p-g_q$, it contributes nothing to the $p$-versus-$q$ numerator and contributes $\eta$ to each denominator.  The candidate gap is therefore
\[
\sum_{j=1}^{k+1}
\frac{R_j}{P_j+Q_jc+\eta},
\]
which converges as $\eta\to0$ to
\[
\varepsilon(-1)^c.
\]

Since there are only finitely many cases, a sufficiently large finite $M$ preserves all
output signs. Thus the construction is realized with finite parameters.
\end{remark}

\section{\texorpdfstring{$k$ heads cannot solve $(k+1)$-hop induction heads}{k heads cannot solve (k+1)-hop induction heads}}
\label{sec:khop}
We now use the same alternating-sum obstruction to show that a one-layer $k$-head attention-only transformer cannot solve the $(k+1)$-hop induction heads task.  The idea is a \emph{reduction}: we encode an arbitrary $(k+1)$-bit parity instance as a $(k+1)$-hop instance, in such a way that the query token is fixed and every input token depends on at most one parity bit.  A hop solver, restricted to these inputs, would then solve $(k+1)$-parity with $k$ heads, which the obstruction forbids.

\paragraph{The $k$-hop task.}  Recall the $m$-hop induction map of Sanford, Hsu, and Telgarsky.  For a string $X = X_1 \cdots X_N$ over a finite alphabet, let $\operatorname{find}^1_X(i) = \max(\{0\} \cup \{\, j \le i : X_{j-1} = X_i \,\})$ (the position immediately after the rightmost earlier occurrence of the token $X_i$, or $0$ if none), and iterate $\operatorname{find}^m_X = \operatorname{find}^1_X \circ \operatorname{find}^{m-1}_X$.  The $m$-hop task is to output, at the final (query) position $i$, the token $X_{\operatorname{find}^m_X(i)}$.

\paragraph{The encoding.}  Over the five-symbol alphabet $\{a, b, 0, 1, ?\}$, map a parity input $x = (x_0, x_1, \ldots, x_k) \in \{0,1\}^{k+1}$ to the length-$(4k{+}3)$ string
\[
E(x) \;=\; y_0\, y_1 \cdots y_k\, ?,
\]
where each block $y_i$ is defined by
\[
y_0 = a\,x_0\,b\,(1{-}x_0), \qquad
y_i = \begin{cases}
a\,x_i\,b\,(1{-}x_i) & 1 \le i < k,\ i \text{ even},\\[2pt]
0\,a\,1\,b \ \text{ if } x_i = 0,\ \text{ else } 0\,b\,1\,a & 1 \le i < k,\ i \text{ odd},
\end{cases}
\]
and the final block is
\[
y_k = \begin{cases}
?\,x_k & k \text{ even},\\
?\,a \ \text{ if } x_k = 0,\ \text{ else } ?\,b & k \text{ odd},
\end{cases}
\]
followed by a single trailing query token $?$.

\paragraph{The encoding computes parity.}
\begin{lemma}\label{lem:hop-parity}
For every $x \in \{0,1\}^{k+1}$, the $(k{+}1)$-hop output of $E(x)$ at the trailing query is the token $\bigoplus_{i=0}^{k} x_i \in \{0,1\}$.
\end{lemma}

\noindent The chase realizes an XOR accumulator.  Hop $1$ leaves the query $?$ and reads $x_k$ through $y_k$, landing in a two-valued ``state'' carried by the pair $\{a,b\}$ (or $\{0,1\}$).  Hops $2, \ldots, k$ read $x_{k-1}, \ldots, x_1$ through $y_{k-1}, \ldots, y_1$; the alternation between the $\{a,b\}$ blocks (even levels) and the $\{0,1\}$ blocks (odd levels) is arranged so that each hop lands on the successor token encoding the running parity $x_i \oplus \cdots \oplus x_k$.  After hop $k$ the state (an $a$ or a $b$) encodes $x_1 \oplus \cdots \oplus x_k$.  Hop $k{+}1$ resolves this state through $y_0 = a\,x_0\,b\,(1{-}x_0)$, which sends $a \mapsto x_0$ and $b \mapsto 1{-}x_0$, so the emitted token is
\[
x_0 \oplus (x_1 \oplus \cdots \oplus x_k) \;=\; \bigoplus_{i=0}^{k} x_i.
\]
(The block $y_0$ is exactly an ``even'' block for level $0$; freezing it to $a\,0\,b\,1$ would hardcode $x_0 = 0$ and lose one parity bit, reducing the reduction to only $k$ free bits.)

\paragraph{Structural properties of the encoding.}  Two facts hold by construction and are all we need:
\begin{itemize}
\item[(P1)] \textbf{The query embedding is fixed.} Let $N=4k+3$, write $v_{t,p}:=\Phi(t,p)$, and set $v_{\mathrm{qry}}:=\Phi(?,N)$. For a head $j$, the score to position $p$ is $v_{\mathrm{qry}}A_jv_{t_p,p}^{\top}$, so its exponential $r_p^j$ depends on the input only through the token $t_p$ at that fixed position.
\item[(P2)] \textbf{Every position depends on at most one bit.} In each block the fixed symbols are constant and every variable slot is a function of just the bit indexing its block. Hence both $r_p^j$ and any fixed scalar projection of $v_{t_p,p}$ depend on at most one of $x_0,\ldots,x_k$.
\end{itemize}

\begin{theorem}\label{thm:khop-lb}
There is no one-layer $k$-head attention-only transformer that solves the $(k{+}1)$-hop induction heads task.
\end{theorem}

\begin{proof}
Suppose such a transformer exists. Run it on the inputs $E(x)$, $x\in\{0,1\}^{k+1}$. Solving the task requires emitting the correct output token at the query; by \Cref{lem:hop-parity} that token is $\bigoplus_{i=0}^k x_i\in\{0,1\}$. Write $y_j^{(t)}$ for the column of $Y_j=V_jW_O^\top$ indexed by output token $t$, set $a_j:=y_j^{(0)}$ and $b_j:=y_j^{(1)}$, and, for the token $t_p$ at position $p$, let
\[
  r_p^j:=e^{v_{\mathrm{qry}}A_jv_{t_p,p}^{\top}}>0,\qquad
  d_x^j:=\sum_{p=1}^{N}r_p^j>0,
  \qquad n_x^j:=\sum_{p=1}^{N}r_p^jv_{t_p,p},
\]
where the sums run over all positions of $E(x)$.

Correctness on the whole family requires, for every $x$, that the logit of the correct symbol among $\{0,1\}$ at least match the logit of the other symbol and, under \Cref{def:k-attn}'s tie-break toward $0$, strictly exceed it whenever the correct symbol is $1$. Writing
\[
  \Delta_{01}(x):=Z_0(x)-Z_1(x)
  =\sum_{j=1}^k\frac{n_x^j}{d_x^j}\cdot(a_j-b_j)
\]
this reads
\[
(-1)^{\sum_{i=0}^kx_i}\Delta_{01}(x)
\;\begin{cases}
>0 & \text{if }\bigoplus_{i=0}^kx_i=1,\\[2pt]
\ge0 & \text{if }\bigoplus_{i=0}^kx_i=0.
\end{cases}
\]

By (P1)--(P2), each summand of $s_x^j:=n_x^j\cdot(a_j-b_j)$ and of $d_x^j$ depends on at most one of the $k+1$ bits $x_0,\ldots,x_k$, and $d_x^j>0$. This is precisely the hypothesis of the parity obstruction (\Cref{rem:parity-obstruction}) with $P=k+1$ binary indices and $H=k$ heads.

Explicitly, multiplying the inequality at $x$ by $\prod_jd_x^j>0$ gives $\sigma_xF_x\ge0$ for every $x$, with strict inequality whenever $\sigma_x=-1$, where
\[
\sigma_x:=(-1)^{\sum_{i=0}^kx_i},\qquad
F_x:=\prod_{j=1}^kd_x^j\,\Delta_{01}(x)
=\sum_{j=1}^ks_x^j\prod_{j'\ne j}d_x^{j'}.
\]
Each monomial of $F_x$ is a product of $k$ summands, each depending on at most one bit by (P1)--(P2), so it depends on at most $k<k+1$ bits. Summing that monomial against $\sigma_x$ cancels over a missing coordinate, giving $\sum_x\sigma_xF_x=0$. But the $2^k$ odd-weight inputs each contribute a strictly positive term and every other term is nonnegative, so $\sum_x\sigma_xF_x>0$, a contradiction.

\end{proof}

\begin{remark}
The bound is one hop short of tight in the following sense.  The reduction shows $(k{+}1)$-hop needs at least $k{+}1$ heads, matching the pattern that $m$-hop induction requires at least $m$ heads in a one-layer attention-only transformer.  This complements the constructive side, where such tasks are solved with a number of heads exponential in the number of hops; closing the gap between the linear lower bound and the exponential upper bound is left open.  The reduction also uses only the necessary condition that the correct $\{0,1\}$ symbol outrank the other, so it applies to any model that solves the task, regardless of how the remaining alphabet logits behave.
\end{remark}

\section{Reducing Joint Embeddings to Additive Ones}
\label{sec:redn2additive}
Since additive embeddings are a restriction of joint embeddings,
$\mathcal{F}_k^{\mathrm{add}}\subseteq\mathcal{F}_k^{\mathrm{joint}}$.  Conversely, we show that every joint $k$-head computation admits an additive simulation using $O(n^k)$ heads.

\subsection{\texorpdfstring{A joint $k$-head net is a degree-$k$ polynomial threshold function}{A joint k-head net is a degree-k polynomial threshold function}}

Use the affine scalar convention of \Cref{def:ltf}. For head $j$, let
$c_j:=V_j\omega$, set $u_{b,i}^j:=v_{b,i}\cdot c_j$ and
$u_?^j:=v_?\cdot c_j$, and retain the attention weights $r$ and denominators
$d_x^j$ of the scalar normal form. Then
\[
  s_x^j:=r_?^ju_?^j+\sum_{i=1}^nr_{x_i,i}^ju_{x_i,i}^j,
  \qquad
  \mathrm{net}(x)=\mathbf1\!\left\{\sum_{j=1}^k\frac{s_x^j}{d_x^j}-\theta>0\right\}.
\]

\begin{lemma}[PTF form]
\label{lem:ptf}
For a joint $k$-head transformer, since every $d^{j}_x>0$, multiplying the decision inequality by $\prod_j d^{j}_x$ is sign-preserving, so $\mathrm{net}(x)=\mathbf 1\{P_x>0\}$ with
\[
  P_x=\sum_{j=1}^{k}s_x^j\!\!\prod_{j'\neq j}\! d^{j'}_x\;-\;\theta\prod_{j=1}^{k}d^{j}_x .
\]
As a function of the input $x$, each $s_x^j$ and $d^{j}_x$ is affine, so $P_x$ has degree $\le k$ in $x$: each summand contains one affine numerator and $k-1$ affine denominators, while the threshold term contains $k$ affine denominators. Reducing modulo $x_i^{2}=x_i$ gives a multilinear polynomial
\[
  \tilde P(x)=\sum_{|S|\le k}\alpha_S\prod_{i\in S}x_i
\]
that agrees with $P_x$ on $\{0,1\}^n$ (so $\mathbf 1\{\tilde P>0\}=\mathbf 1\{P_x>0\}=f$ there), with at most $\binom{n}{\le k}=\sum_{r=0}^{k}\binom{n}{r}$ monomials.
\end{lemma}

This is the containment $\mathcal{F}_k\subseteq$ \{degree-$k$ polynomial threshold functions\}. Note the two \emph{different} degree counts in play: $P_x$ has degree $\le k$ in the \emph{input} $x$ (used here), whereas viewed as a polynomial in the head \emph{parameters} it has degree $\le k+1$ (the form used for the counting lower bound, \Cref{lem:reduction}).

\subsection{The monomial vote, run on \texorpdfstring{$P$}{P}}

The additive model detects a monotone term $\prod_{i\in S}x_i$ sharply (\Cref{lem:monotone-term}). The reduction simply re-votes the joint net over the monotone-monomial basis of its \emph{own} degree-$k$ threshold polynomial.

\begin{theorem}[Joint $\Rightarrow$ additive]
\label{thm:redn}
Suppose $f:\{0,1\}^n\to\{0,1\}$ is computed by a joint $k$-head transformer. Then $f$ is computed by an \emph{additive} transformer with
\[
  k'\ \le\ \min\Bigl\{2^{n},\ \tbinom{n}{\le k}\Bigr\}=O(n^{k})\ \text{heads},\qquad d'=O(n).
\]

\end{theorem}
\begin{proof}
Because the Boolean cube is finite, we may increase $\theta$ by a sufficiently small amount without changing any prediction; this makes every decision strict.
Let $\tilde P=\sum_{|S|\le k}\alpha_S\prod_{i\in S}x_i$ be the multilinear threshold polynomial of \Cref{lem:ptf}, and take one monotone-term head (\Cref{lem:monotone-term}) per nonempty $S$ with $\alpha_S\ne0$, common steepness $R$, read-out weight $\alpha_S$, and threshold folding in $\alpha_\varnothing$; this is exactly the monomial vote of \Cref{thm:monomial-vote}, but applied to $\tilde P$ rather than to the M\"obius expansion of $f$.
Consequently, $\gamma_P:=\min_x|\tilde P(x)|>0$.
The additive read-out is $F(x)=\sum_S\alpha_S A_S(x)$, and by \Cref{lem:monotone-term}
\[
  \bigl|F(x)-\tilde P(x)\bigr|\ \le\ n\,e^{-R/2}\sum_S|\alpha_S| .
\]

Since $\tilde P$ has finitely many coefficients, the right-hand side above tends to $0$ as $R\to\infty$; choose a finite $R$ for which it is smaller than $\gamma_P$. Then $\operatorname{sign}F=\operatorname{sign}\tilde P$, so $\mathbf 1\{F>0\}=f$. Each head has dimension $O(n)$ by \Cref{lem:monotone-term}, and there are at most $\binom{n}{\le k}$ nonempty sets $S$ with $\alpha_S\ne0$.
Therefore, for functions of arity $n$,
\[
  \mathcal{F}_{k}^{\mathrm{joint}}
  \ \subseteq
  \mathcal{F}_{\min\{2^n,\binom{n}{\le k}\}}^{\mathrm{add}}
  \ \subseteq
  \mathcal{F}_{O(n^k)}^{\mathrm{add}}.
\]
\end{proof}

\section*{Acknowledgements}

Generative AI tools were used to improve the clarity and presentation of the writing and proofs. The authors reviewed and verified all resulting content and take full responsibility for the manuscript.

\bibliographystyle{alphaurl}
\bibliography{references}

@inproceedings{vaswani2017attention,
  author    = {Ashish Vaswani and Noam Shazeer and Niki Parmar and Jakob Uszkoreit and Llion Jones and Aidan N. Gomez and {\L}ukasz Kaiser and Illia Polosukhin},
  title     = {Attention Is All You Need},
  booktitle = {Advances in Neural Information Processing Systems (NeurIPS)},
  year      = {2017},
  eprint    = {1706.03762}
}

@article{elhage2021mathematical,
  author  = {Nelson Elhage and Neel Nanda and Catherine Olsson and Tom Henighan and Nicholas Joseph and Ben Mann and Amanda Askell and Yuntao Bai and Anna Chen and Tom Conerly and Nova DasSarma and Dawn Drain and Deep Ganguli and Zac Hatfield-Dodds and Danny Hernandez and Andy Jones and Jackson Kernion and Liane Lovitt and Kamal Ndousse and Dario Amodei and Tom Brown and Jack Clark and Jared Kaplan and Sam McCandlish and Chris Olah},
  title   = {A Mathematical Framework for Transformer Circuits},
  journal = {Transformer Circuits Thread},
  year    = {2021},
  url     = {https://transformer-circuits.pub/2021/framework/index.html}
}

@article{olsson2022induction,
  author  = {Catherine Olsson and Nelson Elhage and Neel Nanda and Nicholas Joseph and Nova DasSarma and Tom Henighan and Ben Mann and Amanda Askell and Yuntao Bai and Anna Chen and Tom Conerly and Dawn Drain and Deep Ganguli and Zac Hatfield-Dodds and Danny Hernandez and Scott Johnston and Andy Jones and Jackson Kernion and Liane Lovitt and Kamal Ndousse and Dario Amodei and Tom Brown and Jack Clark and Jared Kaplan and Sam McCandlish and Chris Olah},
  title   = {In-context Learning and Induction Heads},
  journal = {Transformer Circuits Thread},
  year    = {2022},
  eprint  = {2209.11895}
}

@misc{hsu2026parity,
  author = {Daniel Hsu},
  title  = {Lower Bounds for One-Layer Transformers That Compute Parity},
  year   = {2026},
  eprint = {2605.12171}
}

@misc{kozachinskiy2026parity,
  author = {Alexander Kozachinskiy and Tomasz Steifer and Przemys{\l}aw Wa{\l}{\k{e}}ga},
  title  = {Parity, Sensitivity, and Transformers},
  year   = {2026},
  eprint = {2602.05896}
}

@article{hahn2020theoretical,
  author  = {Michael Hahn},
  title   = {Theoretical Limitations of Self-Attention in Neural Sequence Models},
  journal = {Transactions of the Association for Computational Linguistics},
  volume  = {8},
  pages   = {156--171},
  year    = {2020},
  doi     = {10.1162/tacl_a_00306},
  eprint  = {1906.06755}
}

@inproceedings{chiang2022overcoming,
  author    = {David Chiang and Peter Cholak},
  title     = {Overcoming a Theoretical Limitation of Self-Attention},
  booktitle = {Proceedings of the 60th Annual Meeting of the Association for Computational Linguistics (ACL)},
  year      = {2022},
  doi       = {10.18653/v1/2022.acl-long.527},
  eprint    = {2202.12172}
}

@inproceedings{bhattamishra2020ability,
  author    = {Satwik Bhattamishra and Kabir Ahuja and Navin Goyal},
  title     = {On the Ability and Limitations of Transformers to Recognize Formal Languages},
  booktitle = {Proceedings of the 2020 Conference on Empirical Methods in Natural Language Processing (EMNLP)},
  year      = {2020},
  doi       = {10.18653/v1/2020.emnlp-main.576},
  eprint    = {2009.11264}
}

@inproceedings{tesfaye2026twoheads,
  author    = {Amanuel Tesfaye and Zeno Kujawa and Rajmohan Rajaraman and Ravi Sundaram},
  title     = {Two (Narrow) Heads Are Better Than (an Arbitrarily Wide) One},
  booktitle = {International Conference on Learning Representations (ICLR)},
  year      = {2026},
  url       = {https://openreview.net/pdf?id=RRmPbbZsvl}
}

@inproceedings{sanford2023representational,
  author    = {Clayton Sanford and Daniel Hsu and Matus Telgarsky},
  title     = {Representational Strengths and Limitations of Transformers},
  booktitle = {Advances in Neural Information Processing Systems (NeurIPS)},
  year      = {2023},
  eprint    = {2306.02896}
}

@inproceedings{sanford2024logdepth,
  author    = {Clayton Sanford and Daniel Hsu and Matus Telgarsky},
  title     = {Transformers, Parallel Computation, and Logarithmic Depth},
  booktitle = {Proceedings of the 41st International Conference on Machine Learning (ICML)},
  year      = {2024},
  eprint    = {2402.09268}
}

@inproceedings{ye2026graphconnectivity,
  author    = {Qilin Ye and Deqing Fu and Robin Jia and Vatsal Sharan},
  title     = {Transformers Provably Learn Algorithmic Solutions for Graph Connectivity, But Only with the Right Data},
  booktitle = {International Conference on Machine Learning (ICML)},
  year      = {2026},
  eprint    = {2510.19753}
}

@misc{sanford2024onelayer,
  author = {Clayton Sanford and Daniel Hsu and Matus Telgarsky},
  title  = {One-Layer Transformers Fail to Solve the Induction Heads Task},
  year   = {2024},
  eprint = {2408.14332}
}

@inproceedings{yu2026headcount,
  author    = {Penghao Yu and Haotian Jiang and Zeyu Bao and Ruoxi Yu and Qianxiao Li},
  title     = {The Effect of Attention Head Count on Transformer Approximation},
  booktitle = {International Conference on Learning Representations (ICLR)},
  year      = {2026},
  eprint    = {2510.06662}
}

@misc{strobl2025concise,
  author = {Lena Strobl and Dana Angluin and Robert Frank},
  title  = {Concise One-Layer Transformers Can Do Function Evaluation (Sometimes)},
  year   = {2025},
  eprint = {2503.22076}
}

@inproceedings{peng2024limitations,
  author    = {Binghui Peng and Srini Narayanan and Christos Papadimitriou},
  title     = {On Limitations of the Transformer Architecture},
  booktitle = {First Conference on Language Modeling (COLM)},
  year      = {2024},
  eprint    = {2402.08164}
}

@misc{yehudai2024count,
  author = {Gilad Yehudai and Haim Kaplan and Guy Dar and Royi Rassin and Asma Ghandeharioun and Mor Geva and Amir Globerson},
  title  = {When Can Transformers Count to n?},
  year   = {2024},
  eprint = {2407.15160}
}

@inproceedings{kozachinskiy2025strassen,
  author    = {Alexander Kozachinskiy and Felipe Urrutia and Hector Jimenez and Tomasz Steifer and Germ{\'a}n Pizarro and Mat{\'i}as Fuentes and Francisco Meza and Cristian B. Calderon and Crist{\'o}bal Rojas},
  title     = {Strassen Attention, Split {VC} Dimension and Compositionality in Transformers},
  booktitle = {Advances in Neural Information Processing Systems (NeurIPS)},
  year      = {2025},
  eprint    = {2501.19215}
}

@inproceedings{alman2023fast,
  author    = {Josh Alman and Zhao Song},
  title     = {Fast Attention Requires Bounded Entries},
  booktitle = {Advances in Neural Information Processing Systems (NeurIPS)},
  year      = {2023},
  eprint    = {2302.13214}
}

@inproceedings{alman2025subquadratic,
  author    = {Josh Alman and Hantao Yu},
  title     = {Fundamental Limitations on Subquadratic Alternatives to Transformers},
  booktitle = {International Conference on Learning Representations (ICLR)},
  year      = {2025},
  eprint    = {2410.04271}
}

@inproceedings{yun2020universal,
  author    = {Chulhee Yun and Srinadh Bhojanapalli and Ankit Singh Rawat and Sashank J. Reddi and Sanjiv Kumar},
  title     = {Are Transformers Universal Approximators of Sequence-to-Sequence Functions?},
  booktitle = {International Conference on Learning Representations (ICLR)},
  year      = {2020},
  eprint    = {1912.10077}
}

@inproceedings{kajitsuka2024onelayer,
  author    = {Tokio Kajitsuka and Issei Sato},
  title     = {Are Transformers with One Layer Self-Attention Using Low-Rank Weight Matrices Universal Approximators?},
  booktitle = {International Conference on Learning Representations (ICLR)},
  year      = {2024},
  eprint    = {2307.14023}
}

@inproceedings{dong2021rankcollapse,
  author    = {Yihe Dong and Jean-Baptiste Cordonnier and Andreas Loukas},
  title     = {Attention is Not All You Need: Pure Attention Loses Rank Doubly Exponentially with Depth},
  booktitle = {Proceedings of the 38th International Conference on Machine Learning (ICML)},
  year      = {2021},
  eprint    = {2103.03404}
}

@inproceedings{mahdavi2024memorization,
  author    = {Sadegh Mahdavi and Renjie Liao and Christos Thrampoulidis},
  title     = {Memorization Capacity of Multi-Head Attention in Transformers},
  booktitle = {International Conference on Learning Representations (ICLR)},
  year      = {2024},
  eprint    = {2306.02010}
}

@inproceedings{merrill2025depth,
  author    = {William Merrill and Ashish Sabharwal},
  title     = {A Little Depth Goes a Long Way: The Expressive Power of Log-Depth Transformers},
  booktitle = {Advances in Neural Information Processing Systems (NeurIPS)},
  year      = {2025},
  eprint    = {2503.03961}
}

@article{merrill2022saturated,
  author  = {William Merrill and Ashish Sabharwal and Noah A. Smith},
  title   = {Saturated Transformers are Constant-Depth Threshold Circuits},
  journal = {Transactions of the Association for Computational Linguistics},
  volume  = {10},
  pages   = {843--856},
  year    = {2022},
  doi     = {10.1162/tacl_a_00493},
  eprint  = {2106.16213}
}

@article{merrill2023parallelism,
  author  = {William Merrill and Ashish Sabharwal},
  title   = {The Parallelism Tradeoff: Limitations of Log-Precision Transformers},
  journal = {Transactions of the Association for Computational Linguistics},
  volume  = {11},
  pages   = {531--545},
  year    = {2023},
  eprint  = {2207.00729}
}

@article{strobl2024survey,
  author  = {Lena Strobl and William Merrill and Gail Weiss and David Chiang and Dana Angluin},
  title   = {What Formal Languages Can Transformers Express? {A} Survey},
  journal = {Transactions of the Association for Computational Linguistics},
  volume  = {12},
  pages   = {543--561},
  year    = {2024},
  doi     = {10.1162/tacl_a_00663},
  eprint  = {2311.00208}
}

@article{perez2021turing,
  author  = {Jorge P{\'e}rez and Pablo Barcel{\'o} and Javier Marinkovic},
  title   = {Attention is {T}uring-Complete},
  journal = {Journal of Machine Learning Research},
  volume  = {22},
  number  = {75},
  pages   = {1--35},
  year    = {2021},
  url     = {https://jmlr.org/papers/v22/20-302.html}
}

@inproceedings{liu2023shortcuts,
  author    = {Bingbin Liu and Jordan T. Ash and Surbhi Goel and Akshay Krishnamurthy and Cyril Zhang},
  title     = {Transformers Learn Shortcuts to Automata},
  booktitle = {International Conference on Learning Representations (ICLR)},
  year      = {2023},
  eprint    = {2210.10749}
}

@article{paturi1994rational,
  author  = {Ramamohan Paturi and Michael E. Saks},
  title   = {Approximating Threshold Circuits by Rational Functions},
  journal = {Information and Computation},
  volume  = {112},
  number  = {2},
  pages   = {257--272},
  year    = {1994},
  doi     = {10.1006/inco.1994.1059}
}

@book{minsky1969perceptrons,
  author    = {Marvin Minsky and Seymour Papert},
  title     = {Perceptrons: An Introduction to Computational Geometry},
  publisher = {MIT Press},
  year      = {1969}
}

@article{warren1968,
  author  = {Hugh E. Warren},
  title   = {Lower Bounds for Approximation by Nonlinear Manifolds},
  journal = {Transactions of the American Mathematical Society},
  volume  = {133},
  number  = {1},
  pages   = {167--178},
  year    = {1968}
}

@article{milnor1964,
  author  = {John Milnor},
  title   = {On the {B}etti Numbers of Real Varieties},
  journal = {Proceedings of the American Mathematical Society},
  volume  = {15},
  number  = {2},
  pages   = {275--280},
  year    = {1964}
}

@incollection{thom1965,
  author    = {Ren{\'e} Thom},
  title     = {Sur l'homologie des vari{\'e}t{\'e}s alg{\'e}briques r{\'e}elles},
  booktitle = {Differential and Combinatorial Topology (A Symposium in Honor of Marston Morse)},
  editor    = {Stewart S. Cairns},
  pages     = {255--265},
  publisher = {Princeton University Press},
  year      = {1965}
}

@article{renegar1992,
  author  = {James Renegar},
  title   = {On the Computational Complexity and Geometry of the First-Order Theory of the Reals. {P}art {I}},
  journal = {Journal of Symbolic Computation},
  volume  = {13},
  number  = {3},
  pages   = {255--299},
  year    = {1992}
}

@book{basu2006algorithms,
  author    = {Saugata Basu and Richard Pollack and Marie-Fran{\c{c}}oise Roy},
  title     = {Algorithms in Real Algebraic Geometry},
  edition   = {2},
  series    = {Algorithms and Computation in Mathematics},
  volume    = {10},
  publisher = {Springer},
  year      = {2006}
}

@article{shannon1949,
  author  = {Claude E. Shannon},
  title   = {The Synthesis of Two-Terminal Switching Circuits},
  journal = {The Bell System Technical Journal},
  volume  = {28},
  number  = {1},
  pages   = {59--98},
  year    = {1949}
}

\end{document}